%% file: main.tex
\documentclass[11pt,letterpaper]{article}

\usepackage[utf8]{inputenc}
\usepackage[T1]{fontenc}
\usepackage{lmodern}
\usepackage[DIV=11]{typearea} % large value = less margin (roughly speaking), recommended values: 10pt: 8, 11pt: 10, 12pt: 12
\usepackage{tikz}

\usepackage{microtype}

\usepackage{amssymb}
\usepackage{amsmath}
\usepackage{amsthm}
\usepackage{thmtools}
\usepackage{thm-restate}

\usepackage{xcolor}
\usepackage{xspace}
\usepackage{xfrac}

\usepackage{comment}
\usepackage[normalem]{ulem}

\usepackage[backend=biber, style=alphabetic, backref=true, doi=false, url=false, maxcitenames=3, mincitenames=3, maxbibnames=10, minbibnames=10, sortlocale=en_US]{biblatex}  

\usepackage[ocgcolorlinks]{hyperref} % Option ocgcolorlinks makes links only colored on screen and not on printouts
\usepackage{cleveref}
\usepackage[bottom]{footmisc}

\usepackage[procnumbered,ruled,vlined,linesnumbered]{algorithm2e}

\input{commands}

\DeclareUnicodeCharacter{221A}{$\sqrt{}$}
\DeclareUnicodeCharacter{3F5}{$\eps$}

\colorlet{DarkRed}{red!50!black}
\colorlet{DarkGreen}{green!50!black}
\colorlet{DarkBlue}{blue!50!black}

\hypersetup{
	linkcolor = DarkRed,
	citecolor = DarkGreen,
	urlcolor = DarkBlue,
	bookmarksnumbered = true,
	linktocpage = true
}

\DontPrintSemicolon
\SetKw{KwAnd}{and}
\SetProcNameSty{textsc}
\SetFuncSty{textsc}

\declaretheorem[numberwithin=section]{theorem}
\declaretheorem[numberlike=theorem]{lemma}

\declaretheorem[numberlike=theorem]{definition}
\declaretheorem[numberlike=theorem]{claim}

\declaretheorem[numberlike=theorem]{invariant}
\declaretheorem[numberlike=theorem]{hypothesis}

\DeclareMathOperator*{\argmax}{arg\,max}

\def \eps {\varepsilon}

\title{New Algorithms and Hardness for Incremental Single-Source Shortest Paths in Directed Graphs}
\author{
Maximilian Probst Gutenberg\thanks{
The author is supported by Basic Algorithms Research Copenhagen (BARC), supported by Thorup's Investigator Grant from the Villum Foundation under Grant No. 16582.}~\thanks{Work done while
visiting the Massachusetts Institute of Technology, Massachusetts, US. The author is supported by STIBOFONDEN’s IT Travel Grant for PhD Students.}\\ University of Copenhagen
\and Virginia Vassilevska Williams \thanks{The author is supported by an NSF CAREER Award, NSF
Grants CCF-1528078 and CCF-1514339, a BSF Grant
BSF:2012338, a Sloan Research Fellowship and a Google faculty fellowship.} \\MIT 
\and Nicole Wein  \thanks{The author is supported by an NSF Graduate Fellowship and NSF Grant CCF-1514339.} \\MIT}
\date{}

\hypersetup{
	pdftitle = {incrDiSSSPUpperAndLowerBounds},
	pdfauthor = {Maximilian Probst Gutenberg, Virginia Vassilevska Williams, Nicole Wein}
}

\begin{document}

\maketitle
\thispagestyle{empty}

\begin{abstract}
In the dynamic Single-Source Shortest Paths (SSSP) problem, we are given a graph $G=(V,E)$ subject to edge insertions and deletions and a source vertex $s\in V$, and the goal is to maintain the distance $d(s,t)$ for all $t\in V$. 

Fine-grained complexity has provided strong lower bounds for exact partially dynamic SSSP and approximate fully dynamic SSSP [ESA'04, FOCS'14, STOC'15]. Thus much focus has been directed towards finding efficient partially dynamic $(1+\eps)$-approximate SSSP algorithms [STOC'14, ICALP'15, SODA'14, FOCS'14, STOC'16, SODA'17, ICALP'17, ICALP'19, STOC'19, SODA'20, SODA'20]. 
Despite this rich literature, for directed graphs there are no known deterministic algorithms for $(1+\eps)$-approximate dynamic SSSP that perform better than the classic ES-tree [JACM'81]. We present the first such algorithm. 

We present a \emph{deterministic} data structure for incremental SSSP in weighted directed graphs with total update time $\tilde{O}(n^2 \log W)$ which is near-optimal for very dense graphs; here $W$ is the ratio of the largest weight in the graph to the smallest.
Our algorithm also improves over the best known partially dynamic \emph{randomized} algorithm for directed SSSP by Henzinger et al. [STOC'14, ICALP'15] if $m=\omega(n^{1.1})$. 

Complementing our algorithm, we provide improved conditional lower bounds. Henzinger et al. [STOC'15] showed that under the OMv Hypothesis, the partially dynamic exact $s$-$t$ Shortest Path problem in undirected graphs requires amortized update or query time $m^{1/2-o(1)}$, given polynomial preprocessing time. Under a hypothesis about finding Cliques, we improve the update and query lower bound for algorithms with polynomial preprocessing time to $m^{0.626-o(1)}$. Further, under the $k$-Cycle hypothesis, we show that any partially dynamic SSSP algorithm with $O(m^{2-\eps})$ preprocessing time requires amortized update or query time $m^{1-o(1)}$, which is essentially optimal. All previous conditional lower bounds that come close to our bound [ESA'04,FOCS'14] only held for ``combinatorial'' algorithms, while our new lower bound does not make such restrictions.

\end{abstract}

\clearpage
\pagenumbering{arabic}

\section{Introduction}
\input{intro}

\section{Preliminaries}
%todo should this defn be moved to intro? 
%todo add more prelims here for lower bounds?
%todo say O tilde notation 
%A \emph{partially dynamic} graph in the incremental (resp., decremental) setting is a graph sequence $\mathcal{G} = (G_0, G_1, \dots, G_M)$ on a fixed vertex set $V$ ,
%where the initial graph is $G_0 = (V, \emptyset)$ and each graph $G_i = (V, E_i)$ is obtained from the previous
%graph $G_{i−1}$ in the sequence by adding (resp., deleting) an edge. Our algorithms are for incremental directed graphs and our hardness results are for both incremental/decremental and directed/undirected graphs.
%todo: we don't prove exact log factors..should we?

For a dynamic weighted directed graph $G = (V,E, w)$, we let $G_i$ denote the $i^{th}$ version of $G$ but simply write $G$ if the context is clear. We define $n$ to be the number of vertices in $G$ and we define $m$ to be the maximum number of edges in any version of $G$, respectively. For each vertex $v \in V$, we let the out-neighborhood $\mathcal{N}^{out}(v)$ be the set of all vertices $w$ such that $(v,w) \in E$. Analogously, let $\mathcal{N}^{in}(v) = \{ u \in V | (u,v) \in E\}$. For all $u,v\in V$, let $d(u,v)$ denote the distance from $u$ to $v$. 

For an array $A$, let $A[i,j]$ be the subarray of $A$ from index $i$ to index $j$, inclusive. If $A$ is an array of lists, we define the size of $A[i,j]$ denoted $|A[i,j]|$ to mean the sum $\sum_{k=i}^j|A[k]|$ of the sizes of each list from $A[i]$ to $A[j]$.

We use $\log n$ to mean $\log_2 n$ and for convenience, we assume without loss of generality that $n$ is a power of $2$. For integers $x,y$ we let $\lfloor x \rfloor_y$ denote the largest multiple of $y$ that is at most $x$.

%We further define $\lceil x \rceil_y$ to the the smallest multiple of $y$ larger than $x$. We assume w.l.o.g. that $\sqrt[3]{n}$ is an integer. We further define $\lceil x \rceil_y^-$ to be the value $\lceil x - y \rceil_y$ or equivalently $\lceil x \rceil_y - y$. 

\section{Warm-up: An $O(n^{2+2/3}/\eps)$ Time Algorithm}

In this section we describe an algorithm for incremental SSSP on unweighted directed graphs with total update time  $O(n^{2+2/3}/\eps)$. This algorithm illustrates the main ideas used in our $\tilde{O}(n^2 \log W/\eps)$ algorithm.

\begin{theorem}
There is a deterministic algorithm that given an unweighted directed graph $G=(V,E)$ subject to edge insertions, a vertex $s\in V$, and $\eps>0$, maintains for every vertex $v$ an estimate $\hat{d}(v)$ such that after every update $d(s,v)\leq \hat{d}(v)\leq (1+\eps)d(s,v)$, in total time $O(n^{2+2/3}/\eps)$.
\end{theorem}

To obtain this result, we take inspiration from a simple property of \emph{undirected} graphs: \emph{Any two vertices at distance at least $3$ have disjoint neighborhoods}. This observation is crucial in several spanner/hopset constructions as well as other graph algorithms (for example \cite{awerbuch1985complexity, elkin20041, bodwin2016better, elkin2016hopsets, henzinger2016dynamic}), as well as partially dynamic SSSP on undirected graphs \cite{bernstein2016deterministic, bernstein2017deterministic}. In~\cite{bernstein2016deterministic}, Bernstein and Chechik exploit this property for partially dynamic undirected SSSP in the following way. The property implies that for any vertex $v$ on a given shortest path from $s$ to some $t$, the neighborhood of $v$ is disjoint from almost all of the other vertices on this shortest path. Thus there cannot be too many high-degree vertices on any given shortest path, and therefore high-degree vertices are allowed to induce large additive error which can be exploited to increase the efficiency of the algorithm.

Whilst we would like to argue along the same lines, this property is unfortunately not given in \emph{directed} graphs: {there could be two vertices $u$ and $v$ at distance $3$, and a third vertex $z$ that only has in-coming edges from $u$ and $v$. Clearly, $u$ and $v$ can now still be at distance $3$ whilst their out-neighborhoods overlap}. We overcome this issue by introducing \emph{forward neighborhoods} $\mathcal{FN}(u)$ that only include vertices from the out-neighborhood $\mathcal{N}^{out}(u)$ that are estimated to be further away from the source vertex $s$ than $u$. Now, suppose there are two vertices $u$ and $v$ both appear on some shortest path from $s$ to some $t$ and whose forward neighborhoods overlap. Let $w$ be a vertex in $\mathcal{FN}(u) \cap \mathcal{FN}(v)$. Since $w$ has a larger distance estimate than $u$ and the edge $(u,w)$ is in the graph, the distance estimates of $u$ and $w$ must be close, assuming that each distance estimate does not incur much error.
%Since $w$ is further from $s$ than $u$, we have that the distance estimate of $u$ and $w$ very close since the edge $(u,w)$ can be used to decrease $w$'s distance estimate otherwise. 
Similarly, the distance estimates of $v$ and $w$ must be close. But then the distance estimates of $u$ and $v$ must also be close. Therefore, the forward neighborhood of each vertex on a long shortest path must only overlap with the forward neighborhoods of few other vertices on the path. In summary, our extension of the property to directed graphs is that if the distance estimates of $u$ and $v$ differ by a lot, then $u$ and $v$ have disjoint forward neighborhoods.

\subsection{The data structure}

In order to illustrate our approach, we present a data structure that only maintains approximate distances for vertices $u$ that are at distance $d(s,u) > n^{2/3}$. This already improves the state of the art since  we can maintain the exact distance $d(s,u)$ if $d(s,u)\leq n^{2/3}$ simply by using a classic ES-tree to depth $n^{2/3}$ which runs in time $O(mn^{2/3})$. 

To understand the motivation behind our main idea, let us first consider a slightly modified version of the classic ES-trees that achieves the same running time: We maintain for each vertex $u\in V$ an array $A_u$ with $n$ elements where $A_u[i]$ is the set of all vertices $v \in \mathcal{N}^{out}(u)$ with $d(s,v)=i$. Then, when $d(s,u)$ decreases, the set of vertices in $\mathcal{N}^{out}(u)$ whose distance from $s$ 
%can be decreased using $u$ as a parent 
decreases is exactly the set of vertices stored in $A[d^{NEW}(s,u) + 2, n]$ which we call the \emph{forward neighborhood} $\mathcal{FN}(u)$ of $u$. (Recall that $A[i,j]$ is  the subarray of $A$ from index $i$ to index $j$, inclusive.) Thus, we only need to scan edges with tail $u$ and head $v \in \mathcal{FN}(u)$,
%and since every edge scan results in a distance decrease, it is straightforward to bound the total number of edge scans by $n^2$.
however, we also need to update $A_u$ whenever an in-neighbor of $u$ decreases its distance estimate.

For our data structure which we call a ``lazy'' ES-tree, we relax several constraints and use a lazy update rule. Instead of maintaining the exact value of $d(s,v)$ for all $v \in \mathcal{N}^{out}(u)$, we only maintain an \emph{approximate} distance estimate $\hat{d}(v)$. Whilst we still maintain an array for each vertex $u \in V$, we now only update the position of $v$ only after $\hat{d}(v)$ has decreased by at least $n^{1/3}$ or if $(u,v)$ was scanned by $u$. To emphasize that this array is only updated occasionally, instead of using the notation $A_u$, we use the notation $\texttt{Cache}_u$. Again, we define $\texttt{Cache}_u[\hat{d}(u) + 2, n]$ to be the \emph{forward neighborhood} of $u$ denoted $\mathcal{FN}(u) \subseteq \mathcal{N}^{out}(u)$. Further, if $\mathcal{FN}(u)$ is small (say of size $O(n^{2/3})$), we say $u$ is \emph{light}. Otherwise, we say that $u$ is \emph{heavy}.

Now, we distinguish two scenarios for our update rule: if $u$ is \emph{light}, then we can afford to update the distance estimates of the vertices in $\mathcal{FN}(u)$ after every decrease of $\hat{d}(u)$. However, if $u$ is \emph{heavy}, then we only update the vertices in $\mathcal{FN}(u)$ after the distance estimate $\hat{d}(u)$ has been decreased by at least $n^{1/3}$ since the last scan of $\mathcal{FN}(u)$.

Additionally, for each edge $(u,v)$, every time $\hat{d}(v)$ decreases by at least $n^{1/3}$, we update $v$'s position in $\texttt{Cache}_u$.

Finally, we note that $|\mathcal{FN}(u)|$ changes over time and so we need to define the rules for when a vertex changes from light to heavy and vice versa more precisely. Initially, the graph is empty and we define every vertex to be light. Once the size of $\mathcal{FN}(u)$ is increased to $\gamma = 6n^{2/3}/\eps$, we set $u$ to be heavy. On the other hand, when $|\mathcal{FN}(u)|$ decreases to $\gamma/2$, we set $u$ to be light. Whenever $u$ becomes light, we immediately scan all $v\in \mathcal{FN}(u)$ and decrease each $\hat{d}(v)$ accordingly. This completes the description of our algorithm.

\subsection{Running time analysis}

Let us now analyze the running time of the lazy ES-tree. For each vertex $u$, every time $\hat{d}(u)$ decreases by $n^{1/3}$, we might scan $u$'s entire in- and out-neighborhoods. Since $\hat{d}(u)$ can only decrease at most $n$ times, the total running time for this part of the algorithm is $O(nm /n^{1/3}) = O(mn^{2/3})$. 

For every light vertex $u$, we scan $\mathcal{FN}(u)$ every time $\hat{d}(u)$ decreases. Since $\hat{d}(u)$ can only decrease at most $n$ times and since $u$ is light, the total running time for all vertices spent for this part of the algorithm is $O(\sum_{v \in V} n\gamma )=O(n^{2+2/3}/\eps)$.

Whenever a vertex $u$ changes from heavy to light, we scan $\mathcal{FN}(u)$. If $u$ only changes from heavy to light once per value of $\hat{d}(u)$, then the running time is $O(n^{2+2/3}/\eps)$ by the same argument as the previous paragraph. So, we only consider the times in which $u$ toggles between being light and heavy whilst having the same value of $\hat{d}(u)$. Since the position of vertices in $\texttt{Cache}_u$ can only decrease, the only way for $u$ to become heavy while keeping the same value of $\hat{d}(u)$ is if an edge is inserted. Since $\gamma/2$ edges must be inserted before $u$ becomes heavy since it last became light, there were $\gamma/2$ edge insertions with tail $u$. Since each inserted edge is only added to a single $\mathcal{FN}(u)$ (namely to the forward neighborhood of its tail), we can amortize the cost of scanning the $\gamma/2$ vertices in $\mathcal{FN}(u)$ over the $\gamma/2$ insertions.

Combining everything, and since the classic ES-tree to depth $n^{2/3}$ takes at most $O(mn^{2/3})$ update time when run to depth $n^{2/3}$, we establish the desired running time.

\subsubsection{Analysis of correctness}
Let us now argue that our distance estimates are maintained with multiplicative error $(1+\eps)$. The idea of the argument can be roughly summarized by the following points: 
\begin{enumerate}
    \item the light vertices do not contribute any error,
    \item we can bound the error contributed by pairs of heavy vertices whose forward neighborhoods overlap, and
    \item the number of heavy vertices on any shortest path with pairwise disjoint forward neighborhoods is small.
\end{enumerate}
We point out that while the main idea of allowing large error in heavy parts of the graphs is similar to \cite{bernstein2016deterministic}, we rely on an entirely new method to prove that this incurs only small total error. We start our proof by proving the following useful invariant. 
\begin{invariant}\label{inv:warm}
After every edge update, if $v\in \mathcal{FN}(u)$ then $|\hat{d}(v) - \hat{d}(u)| \leq n^{1/3}$. 
\end{invariant}
\begin{proof}
First suppose that $\hat{d}(u) \leq \hat{d}(v)$. Since $\hat{d}(u)$ and $\hat{d}(v)$ can only decrease, we wish to show that $\hat{d}(u)$ cannot decrease by too much without $\hat{d}(v)$ also decreasing. This is true simply because every time $\hat{d}(u)$ decreases by at least $n^{1/3}$, $\hat{d}(v)$ is set to at most $\hat{d}(u)+1$. 

Now suppose that $\hat{d}(u) > \hat{d}(v)$. Since $\hat{d}(u)$ and $\hat{d}(v)$ can only decrease, we wish to show that $\hat{d}(v)$ cannot decrease by too much while remaining in $\mathcal{FN}(u)$. This is true simply because every time $\hat{d}(v)$ decreases by at least $n^{1/3}$, we update $v$'s position in $\texttt{Cache}_u$. If $\hat{d}(v)< \hat{d}(u)+2$ and $v$'s position in $\texttt{Cache}_u$ is updated, then $v$ leaves $\mathcal{FN}(u)$.
\end{proof}

Consider a shortest path $\pi_{s,t}$ for any $t \in V$, at any stage of the incremental graph $G$. Let $t_0=s$. Then, for all $i$, let $s_{i+1}$ be the first heavy vertex after $t_{i}$ on $\pi_{s,t}$ and let $t_{i+1}$ be the last vertex on $\pi_{s,t}$ whose forward neighborhood intersects with the forward neighborhood of $s_{i+1}$ (possibly $t_{i+1} = s_{i+1}$). Thus, we get pairs $(s_1, t_1), (s_2, t_2), \dots, (s_k, t_k)$. Additionally, let $s_{k+1} = t$. Since the forward neighborhoods of all $s_i$'s are disjoint and of size at least $\gamma/2$ (recall that $s_i$ is heavy), we have that there are at most $k\leq 2n/\gamma$ pairs $(s_i, t_i)$.

For each $i$, let $v_i$ be some vertex in $\mathcal{FN}(s_i)\cap \mathcal{FN}(t_i)$. Note that $v_i$ exists by definition of $t_i$. By Invariant~\ref{inv:warm}, $|\hat{d}(s_i) - \hat{d}(v_i)| \leq n^{1/3}$ and $|\hat{d}(t_i) - \hat{d}(v_i)| \leq n^{1/3}$. Thus, $\hat{d}(t_i) - \hat{d}(s_i) \leq 2n^{1/3}$.

Let $t'_i$ be the vertex on $\pi_{s,t}$ succeeding $t_i$ (except $t'_0 = s$). If $t'_i\in \mathcal{FN}(t_i)$ then by Invariant~\ref{inv:warm}, $\hat{d}(t'_i) - \hat{d}(t_i) \leq n^{1/3}$. Otherwise, $t'_i\not\in \mathcal{FN}(t_i)$ so $\hat{d}(t'_i)\leq \hat{d}(t_i)+1$. So regardless, we have $\hat{d}(t'_i) - \hat{d}(t_i) \leq n^{1/3}$ and therefore, since $\hat{d}(t_i) - \hat{d}(s_i) \leq 2n^{1/3}$, we have $\hat{d}(t'_i) - \hat{d}(s_i) \leq 3 n^{1/3}$.

We will show that if $u$ is a light vertex and $(u,v)$ is an edge, then $\hat{d}(v)\leq \hat{d}(u)+1$. Consider the last of the following events that occurred: a) edge $(u,v)$ was inserted, b) $\hat{d}(u)$ was decremented, or c) $\hat{d}(u)$ became light. In case a), the algorithm decreases $\hat{d}(v)$ to be at most $\hat{d}(u)+1$. In cases b) and c), the algorithm updates the distance estimate of all vertices in $\mathcal{FN}(u)$, so if $\hat{d}(v)>\hat{d}(u)+1$ then $\hat{d}(v)$ is decreased to $\hat{d}(u)+1$. Thus we have shown that $\hat{d}(s_{i+1})-  \hat{d}(t'_i) = d(t'_i, s_{i+1})$. 

Putting everything together, $\pi_{s,t}$ can be partitioned into (possibly empty) path segments $\pi_{s,t}[t'_i, s_{i+1}]$ and $\pi_{s,t}[s_{i+1}, t'_{i+1}]$. Observe that by definition for each path segment $\pi_{s,t}[t'_i, s_{i+1}]$, the vertices of all edge tails on that segment are light. Thus, by preceding arguments, we can now bound $\hat{d}(t)$ by
\[
    \hat{d}(t) \leq \sum_{i=0}^k \hat{d}(s_{i+1}) - \hat{d}(t'_i) + \sum_{i=0}^{k-1} \hat{d}(t'_{i+1}) - \hat{d}(s_{i+1}) < \sum_{i=0}^k d(t'_i, s_{i+1}) + 3k n^{1/3} \leq d(s,t) + n^{2/3}\eps 
\]
The last inequality comes from our bound on $k$ and the definition of $\gamma$. Thus, if $d(s,t)>n^{2/3}$ then $\hat{d}(t) \leq (1+\eps)d(s,t)$. Otherwise, $d(s,t)\leq n^{2/3}$ so the classic ES-tree up to depth $n^{2/3}$ finds the exact value of $d(s,t)$.

\section{An $\tilde{O}(n^2 \log W)$ Update Time Algorithm}

In this section, we describe how to improve the construction above to derive an $\tilde{O}(n^2 \log W/\eps^{2.5})$ algorithm. We first prove the theorem below which gives a $\tilde{O}(n^2/\eps)$ bound for unweighted graphs and later we note that the the data structure can handle weighted graphs using standard edge rounding techniques.

\begin{theorem}[Unweighted version of Theorem \ref{thm:mainResultUpperBound}]
There is a deterministic algorithm that given an unweighted directed graph $G=(V,E)$ subject to edge insertions, a vertex $s\in V$, and $\eps>0$, maintains for every vertex $v$ an estimate $\hat{d}(v)$ such that after every update $d(s,v)\leq \hat{d}(v)\leq (1+\eps)d(s,v)$, and runs in total time $\tilde{O}(n^2/\eps)$. A query for the approximate shortest path from $s$ to any vertex $v$ can be answered in time linear in the number of edges on the path.
\end{theorem}

We point out that our data structure for unweighted graphs is not only near-optimal for $m \sim n^2$, but also near-optimal in terms of $\eps$, since a better polynomial dependency on $\eps$ would imply faster exact algorithms with $\eps \sim 1/n$. 

\subsection{Algorithm overview}
There are two main differences between our $\tilde{O}(n^2/\eps)$ time algorithm and our warm-up\\ $O(mn^{2/3}/\eps)$ time algorithm from the previous section:
\begin{enumerate}

\item Recall that the warm-up algorithm consisted of 1) a classic ES-tree of bounded depth to handle small distances, and 2) a ``lazy'' ES-tree (of depth $n$) to handle large distances. For our $\tilde{O}(n^2/\eps)$ time algorithm we will have $\log n$ ES-trees of varying degrees of laziness and to varying depths where each ES-tree is suited to handle a particular range of distances. In particular, for each $i$ from 0 to $\log n-1$, we have one lazy ES-tree that handles distances between $2^i$ and $2^{i+1}$. The ES-trees that handle larger distances can tolerate more additive error, and are thus lazier. 
%For each lazy ES-tree, if $\tau$ is its depth, let $\hat{d}_{\tau}(v)$ be its distance estimate for vertex $v$.
%larger distances can tolerate more additive error so we can afford to have lazier ES-trees handling larger distances.

\item Recall that in the warm-up algorithm, each vertex $v$ was of one of two types: light or heavy, depending the size of the forward neighborhood $\mathcal{FN}(v)$. For our $\tilde{O}(n^2/\eps)$ time algorithm, each vertex will be in one of $\Theta(\log n)$ heaviness levels. Roughly speaking, a vertex has heaviness $i$ in the lazy ES-tree up to depth $\tau$ if $|\mathcal{FN}(u)|\approx 2^i \frac{n}{ \tau}$. 
%We will also modify the definition of $\mathcal{FN}(u)$, specified later.
\end{enumerate}

Consider one of our $\log n$ lazy ES-trees. Let $\tau$ be its depth and let $\hat{d}_{\tau}(v)$ be its distance estimate for each vertex $v$. A central challenge caused by introducing $\log n$ heaviness levels for each lazy ES-tree is handling the event that a vertex changes heaviness level. We describe why unlike in the warm-up algorithm, handling changes in heaviness levels is not straightforward and requires careful treatment. In the warm-up algorithm, whenever a vertex $u$ changes from heavy to light, we scan all $v\in \mathcal{FN}(u)$ and decrease each $\hat{d}(v)$ accordingly. Then, in the analysis of the warm-up algorithm, we argued that if $u$ only changes from heavy to light once per value of $\hat{d}(u)$, we get the desired running time. Now that we have many heaviness levels and we are aiming for a running time of $\tilde{O}(n^2/\eps)$, we can no longer allow each vertex to change heaviness level every time we decrement $\hat{d}_{\tau}(u)$. In particular, suppose we are analyzing a lazy ES-tree up to depth $D$. Suppose for each vertex $u$, every time we decrement $\hat{d}_{\tau}(u)$, we change $u$'s heaviness level and scan $\mathcal{FN}(u)$ as a result. Then since $|\mathcal{FN}(u)|$ could be $\Omega(n)$, the final running time would be $\Omega(n^2D)$, which is too large. Thus, unlike in the warm-up algorithm, we require that the heaviness of each vertex does not change too often. 

Without further modification of the algorithm, the heaviness level of a vertex $u$ can change a number of times in succession. Suppose each index of  $\texttt{Cache}_u$ from index $\hat{d}_{\tau}(u)-\log n+2$ to index $\hat{d}_{\tau}(u)+1$ contains many vertices such that each of the next $\log n$ times we decrement $\hat{d}_{\tau}(u)$, $\mathcal{FN}(u)$ increases by enough that $u$ increases heaviness level upon each decrement of $\hat{d}_{\tau}(u)$. We would like to forbid $u$ from changing heaviness levels so frequently. To address this issue, we \emph{change the definition} of the forward neighborhood $\mathcal{FN}(u)$. 

In particular, if $\texttt{Cache}_u$ contains many vertices in the set of indices that closely precede $\texttt{Cache}_u[\hat{d}_{\tau}(u)]$, we \emph{preemptively} add these vertices to $\mathcal{FN}(u)$. In the above example, instead of increasing the heaviness of $u$ for every single decrement of $\hat{d}_{\tau}(u)$, we would preemptively increase the heaviness of $u$ by a lot to avoid increasing its heaviness again in the near future. Roughly speaking, vertex $u$ has heaviness $h(u)$ if $h(u)$ is the maximum value such that there are $\sim\frac{2^{h(u)} n}{\tau}$ vertices in $\texttt{Cache}_u[\hat{d}_{\tau}(u)-2^{h(u)},\tau]$. (Note that this definition of heaviness is an oversimplification for the sake of clarity.)

Like in the warm-up algorithm, the heaviness level of a vertex $u$ determines how often we scan $\mathcal{FN}(u)$. If a vertex $u$ has heaviness $h(u)$, this means that we scan $\mathcal{FN}(u)$ whenever the value of $\hat{d}_{\tau}(u)$ becomes a multiple of $2^{h(u)}$. 

In summary, when we decrement $\hat{d}_{\tau}(u)$, the algorithm does roughly the following:
\begin{itemize}
    \item If the value of $\hat{d}_{\tau}(u)$ is a multiple of $2^{h(u)}$, scan all $v\in \mathcal{FN}(u)$ and decrement $\hat{d}_{\tau}(v)$ if necessary.
    \item If the value of $\hat{d}_{\tau}(u)$ is a multiple of $2^{h(u)}$, increase the heaviness of $u$ if necessary.
    %$2^{i}$ for some $i$, increase the heaviness of $u$ up to $i$ if necessary. 
    \item Regardless of the value of $\hat{d}_{\tau}(u)$, check if $u$ has left the forward neighborhood of any other vertex $w$, and if so, decrease the heaviness of $w$ if necessary.
\end{itemize}

\subsection{The Data Structure}

For each number $\tau$ between 1 and $n$ such that $\tau$ is a power of 2, we maintain a ``lazy ES-tree'' data structure $\mathcal{E}_\tau$. The guarantee of the data structure $\mathcal{E}_\tau$ is that for each vertex $v \in V$ with $d(s,v)\in [\tau, 2\tau)$, the estimate $\hat{d}_{\tau}(v)$ maintained by $\mathcal{E}_\tau$ satisfies  $d(s,v)\leq \hat{d}_{\tau}(v) \leq (1+\eps)d(s, v)$. Let $\tau_{max}=2\tau(1+\eps)$. Since $\mathcal{E}_{\tau}$ does not need to provide a $(1+\eps)$-approximation for distances $d(s,v)> 2\tau$, the largest distance estimate maintained by $\mathcal{E}_{\tau}$ is at most $\tau_{max}$. We use the distance estimate $\tau_{max}+1$ for all vertices that do not have distance estimate at most $\tau_{max}$. For all $u\in V$, the final distance estimate $\hat{d}(u)$ is the minimum distance estimate $\hat{d}_{\tau}(u)$ over all data structures $\mathcal{E}_\tau$, treating each $\tau_{max}+1$ as $\infty$.

\paragraph{Definitions.}

We begin by making precise the definitions and notation from the algorithm overview section. For each data structure $\mathcal{E}_\tau$ and for each vertex $u\in V$ we define the following:
\begin{itemize}
\item $\hat{d}_{\tau}(u)$ is the distance estimate maintained by the data structure $\mathcal{E}_\tau$.
    \item $\texttt{Cache}_u$ is an array of 
$\tau_{max}$ lists of vertices whose purpose is to store (possibly outdated) information about $\hat{d}_{\tau}(v)$ for all $v\in \mathcal{N}_{out}(u)$. Every time we update the position of a vertex $v\in \mathcal{N}_{out}(u)$ in $\texttt{Cache}_u$, we move $v$ to $\texttt{Cache}_u[\hat{d}_{\tau}(v)]$.
    \item $h(u)$ is the \emph{heaviness} of $u$. Intuitively, if $u$ has large heaviness, this means that $u$ has a large \emph{forward neighborhood} (defined later) and that we scan $u$'s forward neighborhood infrequently.
    \item $\texttt{CacheIndex}(u)=\lfloor \hat{d}_{\tau}(u)-1 \rfloor_{2^{h(u)}}$. (Recall that $\lfloor x \rfloor_y$ is the largest multiple of $y$ that is at most $x$.) The purpose of $\texttt{CacheIndex}(u)$ is to define the forward neighborhood of $u$, which we do next.
    %We only scan the forward neighborhood of $u$ whenever $\hat{d}_{\tau}(u)$ is a multiple of $h(u)$. The purpose of $\texttt{CacheIndex}(u)$ is simply to keep track of the next multiple of $h(u)$ that $\hat{d}_{\tau}(u)$ will decrease to.
    \item The \emph{forward neighborhood} of $u$, denoted $\mathcal{FN}(u)$ is defined as the the set of vertices in $\texttt{Cache}_u[\texttt{CacheIndex}(u),\tau_{max}]$. Note that $\mathcal{FN}(u)$ is defined differently from the warm-up algorithm due to reasons described in the algorithm overview section.
  %  \item $\texttt{ExpireIndex}(u)=\texttt{CacheIndex}(u)-2^{h(u)}$. The purpose of $\texttt{ExpireIndex}(u)$ is to signify that vertices $v$ with $\hat{d}_{\tau}(v)\leq \texttt{ExpireIndex}(u)$ cannot be in $\mathcal{FN}(u)$. That is, if $\hat{d}_{\tau}(v)\geq \texttt{CacheIndex(u)}$ then $v\in \mathcal{FN}(u)$, if $\hat{d}_{\tau}(v)\leq \texttt{ExpireIndex(u)}$ then $v\not\in\mathcal{FN}(u)$, and if $\texttt{ExpireIndex(u)}<\hat{d}_{\tau}(v)<\texttt{CacheIndex}(u)$ then $v$ may or may not be in $\mathcal{FN}(u)$.
    %That is, if for some vertex $v\in \mathcal{FN}(u)$, $\hat{d}_{\tau}(v)$ decreases to $\texttt{ExpireIndex}(u)$, we move $v$ to $\texttt{Cache}_u(\hat{d}_{\tau}(v))$ so that $v$ is no longer in $\mathcal{FN}(u)$.
    \item $\texttt{Expire}_u$ is an array of $\tau_{max}$ lists of vertices whose purpose is to ensure that $u$ leaves $\mathcal{FN}(v)$ once $\hat{d}_{\tau}(u)$ becomes less than $\texttt{CacheIndex}(v)$. In particular, $v\in \texttt{Expire}_u[i]$ if $u\in \mathcal{FN}(v)$ and $\texttt{CacheIndex}(v)=i$.
    \item We also define $\texttt{CacheIndex}$ with a second parameter, which will be useful for calculating the heaviness of vertices. Let $\texttt{CacheIndex}(v,2^i)=\lfloor \hat{d}_{\tau}(v)-1 \rfloor_{2^i}$. Note that $\texttt{CacheIndex}(u,2^{h(u)})$ is the same as $\texttt{CacheIndex}(u)$.
\end{itemize}

\paragraph{Initialization.}
%The pseudocode for initializing the data structure $\mathcal{E}_{\tau}$ is given in Algorithm~\ref{alg:init}.  
We assume without loss of generality that the initial graph is the empty graph. To initialize each $\mathcal{E}_\tau$, we initialize $\hat{d}_{\tau}(s)$ to 0, and for each $u\in V\setminus \{s\}$, we initialize $\hat{d}_{\tau}(u)$ to $\tau_{max} + 1$. Additionally, for each $u\in V\setminus \{s\}$ we initialize the heaviness $h(u)$ to $0$, and we initialize the arrays $\texttt{Cache}_u$ and $\texttt{Expire}_u$ by setting each of the $\tau_{max} + 1$ fields in each array to an empty list.

% \begin{algorithm}
% \caption{Initialization of the data structure $\mathcal{E}_\tau$.}
% \label{alg:init}
% \SetKwProg{procedure}{Procedure}{}{}
% \procedure{$\textsc{Init}()$}{
%     $\hat{d}_{\tau}(s) \gets 0$\;
%     \ForEach{$u \in V \setminus \{s\}$}
%     {$\hat{d}_{\tau}(u) \gets \tau_{max} + 1$\; $h(u) \gets 0$\;
%         Initialize each of the $\tau_{max} + 1$ fields in $\texttt{Cache}_u$ and $\texttt{Expire}_u$ with an empty list\;
%     }
% }
% \end{algorithm}
%commented out the pseudocode bc I'm not sure it's necessary for the reader

\paragraph{The edge update algorithm.}

The pseudocode for the edge update algorithm is given in Algorithm~\ref{alg:insertEdge1}. We also outline the algorithm in words.

The procedure $\textsc{InsertEdge}(u,v)$ begins by updating $\texttt{Cache}_u$ and $\texttt{Expire}_v$ to reflect the new edge. Then, it calls $\textsc{IncreaseHeaviness}(u)$ to check whether the heaviness of $u$ needs to increase due to the newly inserted edge. Then, it initializes a set $H$ storing edges. 

Initially $H$ contains only the edge $(u,v)$. The purpose of $H$ is to store edges $(x,y)$ after the distance estimate $\hat{d}_{\tau}(x)$ has changed. We then extract one edge at a time and check whether the decrease in $x$'s distance estimate also translates to a decrease of $y$'s distance estimate by checking whether $\hat{d}_{\tau}(y) > \hat{d}_{\tau}(x)+1$. If so, then $\hat{d}_{\tau}(y)$ can be decremented and we keep the edge in $H$. Otherwise, we learned that $(x,y)$ cannot be used to decrease $\hat{d}_{\tau}(y)$ and we remove $(x,y)$ from $H$. We point out that in our implementation a decrease of $\Delta$ is handled in the form of $\Delta$ decrements where the edge is extracted from $H$ $\Delta + 1$ times until it is removed from $H$.

\begin{algorithm}
{
\fontsize{10}{10}\selectfont
\caption{Algorithm for handling edge updates.}
\label{alg:insertEdge1}
\SetKwProg{procedure}{Procedure}{}{}
\procedure{$\textsc{InsertEdge}(u, v)$}{
    Add $v$ to $\texttt{Cache}_u[\hat{d}_{\tau}(v)]$\;\label{line:insertcache}
    \If{$\hat{d}_{\tau}(v) \geq \normalfont{\texttt{CacheIndex}}(u)$}{
        Add $u$ to $\texttt{Expire}_v[\texttt{CacheIndex}(u)]$
    }
    $\textsc{IncreaseHeaviness}(u)$\;
    
    \If{$\hat{d}_{\tau}(v)>\hat{d}_{\tau}(u)+1$ \label{line:ifThenhLoop}}{
        Let $H$ be a set storing edges $(x,y)$ \;
        $H.\textsc{Insert}(u,v)$\; \label{line:Hinsert1}
        \While{$H \neq \emptyset$}{\label{line:while}
            Let tuple $(x,y)$ be any tuple in $H$\;
            \If{$\hat{d}_{\tau}(y) > \hat{d}_{\tau}(x)+1$ \label{line:ifThenDecrement}}{
                $\textsc{Decrement}(x,y)$\;\label{line:calldecrement}
            }\Else{
                $H.\textsc{Remove}(x,y)$
            }
        }
    }
}

\label{alg:decrease1}
\SetKwProg{procedure}{Procedure}{}{}
\procedure{$\textsc{Decrement}(u, v)$}{
    $\hat{d}_{\tau}(v) = \hat{d}_{\tau}(v) - 1$\;

    \If{$\hat{d}_{\tau}(v)$ \normalfont{ is a multiple of } $2^{h(v)}$}
    {
        $\textsc{IncreaseHeaviness}(v)$\;
        \ForEach{$w \in \mathcal{FN}(v)$\label{line:fndecrement}}{
            Move $w$ to $\texttt{Cache}_v[\hat{d}_{\tau}(w)]$\;
            Move $v$ to $\texttt{Expire}_w[\texttt{CacheIndex}(v)]$\;
            $H.\textsc{Insert}(v,w)$\label{line:Hinsert2}
        }
    }
    \ForEach{$w\in \normalfont{\texttt{Expire}}_v[\hat{d}_{\tau}(v)+1]$\label{line:exp}}{
        Move $v$ to $\texttt{Cache}_w[\hat{d}_{\tau}(v)]$\;\label{line:begin}
        Remove $w$ from $\texttt{Expire}_v$\;
        $\textsc{DecreaseHeaviness}(w)$\label{line:end}
    }
}

\label{alg:incheav1}
\SetKwProg{procedure}{Procedure}{}{}
\procedure{$\textsc{IncreaseHeaviness}(u)$}{
    $i' \gets \argmax_{i \in \mathbb{N}}\{| \texttt{Cache}_u[\texttt{CacheIndex}(u,2^i), \tau_{max}]| \geq (2^i-1) \frac{12 n \log n }{\eps\tau}\} $\label{line:calc1}\;

    \If{$i' > h(u)$\label{line:ifloop1}}{
        \ForEach{$v\in \normalfont{\texttt{Cache}}_u[\texttt{CacheIndex}(u,2^{i'}), \tau_{max}]$\label{line:loop1}}{
            Move $v$ to $\texttt{Cache}_u[\hat{d}_{\tau}(v)]$\;
            Remove $u$ from $\texttt{Expire}_v$\;
        }

        $h(u) \gets \argmax_{i \leq i'}\{| \texttt{Cache}_u[\texttt{CacheIndex}(u,2^i), \tau_{max}]| \geq  (2^i-1) \frac{6 n \log n }{\eps\tau}\}$\label{line:calc2} \;
        
        \ForEach{$v\in \mathcal{FN}(u)$\label{line:loop2}}{
            Add $u$ to $\texttt{Expire}_v[\texttt{CacheIndex}(u)]$\;
        }
    }
}

\SetKwProg{procedure}{Procedure}{}{}
\procedure{$\textsc{DecreaseHeaviness}(u)$}{
    $i' \gets \argmax_{i \in \mathbb{N}}\{| \texttt{Cache}_u[\texttt{CacheIndex}(u,2^i), \tau_{max}]| \geq (2^i-1) \frac{6 n \log n }{\eps\tau} \}$\;\label{line:i'dec}

    \If{$i' < h(u)$\label{line:preloop3}}{
        \ForEach{$v\in \mathcal{FN}(u)$\label{line:loop3}}{
            Move $v$ to $\texttt{Cache}_u[\hat{d}_{\tau}(v)]$\;
            Remove $u$ from $\texttt{Expire}_v$\;
        }

        $h(u)\gets \argmax_{i \in \mathbb{N}}\{|  \texttt{Cache}_u[\texttt{CacheIndex}(u,2^i), \tau_{max}]| \geq  (2^i-1) \frac{6 n \log n }{\eps\tau}\}$\label{line:sethd}\;
 
        \ForEach{$v\in \mathcal{FN}(u)$\label{line:loop4}}{
            Add $u$ to $\texttt{Expire}_v[\texttt{CacheIndex}(u)]$\;
            $H.\textsc{Insert}(u,v)$\label{line:Hdec}
        }
    }
}
}
\end{algorithm}

The procedure $\textsc{Decrement}(u,v)$ begins by decrementing $\hat{d}_{\tau}(v)$. Then, it checks whether $\hat{d}_{\tau}(v)$ is a multiple of $2^{h(v)}$. If so, it calls $\textsc{IncreaseHeaviness}(v)$ to check whether the recent decrements of $\hat{d}_{\tau}(v)$ have caused $\mathcal{FN}(v)$ to increase by enough that the heaviness $h(v)$ has increased. Also, if $\hat{d}_{\tau}(v)$ is a multiple of $2^{h(v)}$, $\texttt{CacheIndex}(v)$ and thus $\mathcal{FN}(v)$ have changed. Thus, we scan each vertex $w \in \mathcal{FN}(v)$ and update the position of $w$ in $\texttt{Cache}_v$. Then, we insert for each such vertex $w \in \mathcal{FN}(v)$ the edge $(v,w)$ into $H$ which has the eventual effect of decreasing $\hat{d}_{\tau}(w)$ to value at most $\hat{d}_{\tau}(v) + 1$. Since we perform these actions every $2^{h(u)}$ decrements of $\hat{d}_{\tau}(v)$, as we show later, we incur roughly $2^{h(u)}$ additive error on each out-going edge of $v$.
%since after $2^{h(u)}$ distance decreases, we check whether and how much $w$ can be decreased due to decreases of $v$'s distance estimate. 

Additionally, the procedure $\textsc{Decrement}(u,v)$ checks whether decrementing $\hat{d}_{\tau}(v)$ has caused $v$ to expire from any of the forward neighborhoods that contain $v$. The vertices whose forward neighborhood $v$ needs to leave are stored in $\texttt{Expire}_v[\hat{d}_{\tau}(v)+1]$. For each $w\in \texttt{Expire}_v[\hat{d}_{\tau}(v)+1]$, we update $v$'s position in $\texttt{Cache}_w$ which causes $v$ to leave $\mathcal{FN}(w)$. Then, we call $\textsc{DecreaseHeaviness}(u)$ to check whether removing $v$ from $\mathcal{FN}(w)$ has caused the heaviness of $w$ to decrease. 

The procedures $\textsc{IncreaseHeaviness}(u)$ and $\textsc{DecreaseHeaviness}(u)$ are similar. We first describe $\textsc{DecreaseHeaviness}(u)$. On line~\ref{line:sethd} in $\textsc{DecreaseHeaviness}(u)$, $h(u)$ is set to $\argmax_{i \in \mathbb{N}}\{|  \texttt{Cache}_u[\texttt{CacheIndex}(u,2^i), \tau_{max}]| \geq  (2^i-1) \frac{6 n \log n }{\eps\tau}\}$. We note that $\texttt{Cache}_u$ may contain out-of-date information when $\textsc{DecreaseHeaviness}(u)$ is called, however, we wish to update $h(u)$ based on up-to-date information. Thus, before line~\ref{line:sethd}, we update $\texttt{Cache}_u$. However, we do not have time to update \emph{every} index of $\texttt{Cache}_u$, so instead we only update the relevant indices. To do so, it suffices to first calculate the value $i'$, which is the expression for $h(u)$ but using the out-of-date version of $\texttt{Cache}_u$, and then scan all $v\in \normalfont{\texttt{Cache}}_u[\texttt{CacheIndex}(u,2^{i'}), \tau_{max} ]$, updating the position of each such $v$ in $\texttt{Cache}_u$.

Recall that a smaller value of $h(u)$ means that we scan $\mathcal{FN}(u)$ more often. Thus, after we decrease $h(u)$ in $\textsc{DecreaseHeaviness}(u)$, the vertices $v\in \mathcal{FN}(u)$ might not have been scanned recently enough according to the new value of $h(u)$. Thus, to conclude the procedure $\textsc{DecreaseHeaviness}(u)$, we scan each $v\in \mathcal{FN}(u)$ and add $(u,v)$ to the set $H$ so that $\textsc{Decrement}(u,v)$ is called later.

The main difference between $\textsc{IncreaseHeaviness}(u)$ and $\textsc{DecreaseHeaviness}(u)$ is that the constants in the expressions for calculating $i'$ and $h(u)$ are different from each other, which ensures that $u$ does not change heaviness levels too often. Additionally, the last step of  $\textsc{DecreaseHeaviness}(u)$ where we insert into $H$ is not necessary for $\textsc{IncreaseHeaviness}(u)$.
%might decrease $h(u)$, we might reduce the error that $u$ is allowed on its out-going edges, which is the reason why we insert edges $(u,v)$ into the set $H$ after we decreased $h(u)$.

\subsection{Analysis of correctness}

For each vertex $t$, the algorithm obtains the distance estimate $\hat{d}(t)$ by taking the minimum $\hat{d}_{\tau}(t)$ over all $\tau$ (excluding when $\hat{d}_{\tau}(t)=\tau_{max}$+1). The goal of this section, is to prove that 
\[
d(s,t) \leq \hat{d}(t) \leq (1+\eps)d(s,t)
\]
for $d(s,t) \in [\tau, 2\tau]$. We prove this statement in two steps starting by giving a lower bound on $\hat{d}(t)$. 

\begin{lemma}\label{lem:correct1}
At all times, for all $\tau$, for any $t \in V$, we have $d(s,t) \leq \hat{d}_{\tau}(t)$.
\end{lemma}
\begin{proof}
It suffices to show that we only decrement $\hat{d}_{\tau}(v)$ if $v$ has an in-coming edge from a vertex with distance estimate more than 1 below $\hat{d}_{\tau}(v)$.
We only invoke the procedure $\textsc{Decrement}(u,v)$ from line \ref{line:calldecrement}, and we invoke it under the condition that $(u,v)$ is an edge and $\hat{d}_{\tau}(v) > \hat{d}_{\tau}(u) + 1$. Therefore after running $\textsc{Decrement}(u,v)$ we still have $\hat{d}_{\tau}(v) \geq \hat{d}_{\tau}(u) + 1$. 
%Since we only add tuples $(u,v)$ to $H$ that correspond to unit weight edges in $G$, the lemma follows. 
\end{proof}

Let us next prove a small, but helpful lemma.

\begin{lemma}\label{lem:cachedecreases}
For all vertices $u,v\in V$, the index of $\normalfont{\texttt{Cache}_u}$ containing $v$ can only decrease over time.
\end{lemma}
\begin{proof}
Whenever we insert $v$ into to $\texttt{Cache}_u$ or move $v$ to a new index in $\texttt{Cache}_u$, $v$ is placed in $\texttt{Cache}_u[\hat{d}_{\tau}(v)]$. Since $\hat{d}_{\tau}(v)$ is monotonically decreasing over time, the lemma follows.
\end{proof}

Now, before giving an upper bound on the stretch of the distance estimate, we prove the following invariant which is analogous to Invariant~\ref{inv:warm} from the warm-up algorithm.

\begin{invariant}\label{inv:correct}
For all $u,v\in V$, after processing each edge update, if $v\in \mathcal{FN}(u)$ then $|\hat{d}_{\tau}(v) - \hat{d}_{\tau}(u)| \leq 2^{h(u)}$.
\end{invariant}
\begin{proof} 
%We prove the contrapositive and distinguish by cases. We start by showing that if $\hat{d}_{\tau}(v) < \texttt{CacheIndex}(u)$ 
%(retell that $2^{h(u)} \geq \hat{d}_{\tau}(u) - \texttt{CacheIndex}(u)$)
%, then $v \not\in \mathcal{FN}(u)$. 
We first note that the invariant is initially satisfied since $\mathcal{FN}(u)$ is initially empty.
First we prove that there is no $v \in \mathcal{FN}(u)$ with $\hat{d}_{\tau}(u) - \hat{d}_{\tau}(v) > 2^{h(u)}$. We first note that if $v\in \texttt{Cache}_u[\hat{d}_{\tau}(v)]$, then this inequality holds simply from the definitions of $\mathcal{FN}$ and $\texttt{CacheIndex}$. Thus, it suffices to show that if an event occurs that could potentially cause the inequality to be violated, then we have $v\in \texttt{Cache}_u[\hat{d}_{\tau}(v)]$.
We point out that the inequality could only be violated due to three events:
\begin{enumerate}
    \item \uline{$v \in \mathcal{FN}(u)$ and $\hat{d}_{\tau}(v)$ decreases:} We observe that when $\hat{d}_{\tau}(v)$ decrements, we iterate through each vertex $w \in \normalfont{\texttt{Expire}}_v[\hat{d}_{\tau}(v)+1]$ (line~\ref{line:exp}). Since we update $\texttt{Expire}_v$ immediately after $v$ is moved in $\texttt{Cache}_u$, we have that if $\hat{d}_{\tau}(v)=\texttt{CacheIndex}(u)-1$ then $u\in \normalfont{\texttt{Expire}}_v[\hat{d}_{\tau}(v)+1]$. Thus, if $\hat{d}_{\tau}(v)$ decrements to $\texttt{CacheIndex}(u)-1$, then the loop on line~\ref{line:exp} moves $v$ to $\texttt{Cache}_u[\hat{d}_{\tau}(v)]$.
    %and therefore $v$ is no longer in $\mathcal{FN}(u)$.
    \item %\uline{$\texttt{CacheIndex}(u)$ increases:} We first note that $\texttt{CacheIndex}(u)$ only increases if 
    \uline{$h(u)$ decreases:} We note that only the procedure $\textsc{DecreaseHeaviness}(u)$ can decrease $h(u)$.
    (In particular, $h(u)$ cannot decrease in $\textsc{IncreaseHeaviness}(u)$ by Lemma~\ref{lem:sizefn}.) In $\textsc{DecreaseHeaviness}(u)$, $i'$ and $h(u)$ are each set to the expression 
    \[
    \argmax_{i \in \mathbb{N}}\{| \texttt{Cache}_u[\texttt{CacheIndex}(u,2^i), \tau_{max}]| \geq (2^i-1) \frac{6 n \log n }{\eps\tau} \}
    \]
    on lines \ref{line:i'dec} and \ref{line:sethd}, respectively. Between these two lines, $\hat{d}_{\tau}(u)$ remains fixed, and thus $\texttt{CacheIndex}(u,2^i)$ also remains fixed for all $i$. Between the lines \ref{line:i'dec} and \ref{line:sethd}, we move each vertex $y$ in $\texttt{Cache}_u[\texttt{CacheIndex}(u,2^{i'}), \tau_{max} ]$ to $\texttt{Cache}_u[\hat{d}_{\tau}(y)]$. By Lemma \ref{lem:cachedecreases} this can only decrease the indices of vertices in $\texttt{Cache}_u$ and therefore the size of  $\normalfont{\texttt{Cache}}_u[\texttt{CacheIndex}(u,2^{i'}), \tau_{max}]$ can only decrease. Thus, when we pick the new $h(u)$, it satisfies $h(u) \leq i'$. It follows that each vertex $y\in \mathcal{FN}(u)$
    %= \normalfont{\texttt{Cache}}_u[\texttt{CacheIndex}(u,2^{h(u)}), \tau_{max}]$
    has been moved to $\texttt{Cache}_u[\hat{d}_{\tau}(y)]$.
    %distance estimate at least $\texttt{CacheIndex}(u,2^{h(u)})$. The desired inequality follows from the definition of $\texttt{CacheIndex}$.
    
    \item \uline{$v$ is added to $\mathcal{FN}(u)$:} A vertex $v$ can be added to $\mathcal{FN}(u)$ if either the edge $(u,v)$ is inserted, the distance $\hat{d}_{\tau}(u)$ decreases to a multiple of $2^{h(u)}$, or $h(u)$ increases. If the edge $(u,v)$ is inserted then $v$ is added to $\texttt{Cache}_u[\hat{d}_{\tau}(v)]$ on line \ref{line:insertcache}. If $\hat{d}_{\tau}(u)$ decreases to a multiple of $2^{h(u)}$ then in the loop on line \ref{line:fndecrement}, if $y\in\mathcal{FN}(u)$ then $y$ is moved to $\texttt{Cache}_u[\hat{d}_{\tau}(y)]$. It remains to argue about the last case, where $h(u)$ is increased: we observe that in procedure $\textsc{IncreaseHeaviness}(u)$, we first pick a new potential heaviness $i'$ on line \ref{line:calc1} and then scan all vertices in $\texttt{Cache}_u[\texttt{CacheIndex}(u, 2^{i'}), \tau_{max}]$, moving each vertex $y$ to $\texttt{Cache}_u[\hat{d}_{\tau}(y)]$. Then, we take the new value $h(u) \leq i'$ in line \ref{line:calc2} and since we choose $h(u)$ among values smaller than $i'$, 
    each vertex $y\in \mathcal{FN}(u)$
    has been moved to $\texttt{Cache}_u[\hat{d}_{\tau}(y)]$.
    %each vertex $\mathcal{FN}(u)$ has distance estimate at least $\texttt{CacheIndex}(u,2^{h(u)})$.
    %and therefore each vertex $v$ in the new $\mathcal{FN}(u)$ was scanned previously. 
\end{enumerate}

It remains to prove that there is no $v \in \mathcal{FN}(u)$ with $\hat{d}_{\tau}(v) - \hat{d}_{\tau}(u) > 2^{h(u)}$. Again, we point out that the inequality could only be violated due to three events:
\begin{enumerate}
    \item \uline{$h(u)$ decreases:} Again, only the procedure $\textsc{DecreaseHeaviness}(u)$ can decrease $h(u)$. In line \ref{line:Hdec} in $\textsc{DecreaseHeaviness}(u)$, for every vertex $v \in \mathcal{FN}(u)$ that could potentially have its distance estimate decreased, $(u,v)$ is inserted into the set $H$, which has the eventual effect that $\hat{d}_{\tau}(v)\leq\hat{d}_{\tau}(u)+1$, once $H$ is empty.
    \item \uline{$\hat{d}_{\tau}(u)$ is decremented:} Let $h^{NEW}(u)$ be the value of $h(u)$ at the point in time when we have just decremented $\hat{d}_{\tau}(u)$. Let $\ell$ be the smallest multiple of $2^{h^{NEW}(u)}$ that is at least $\hat{d}_{\tau}(u)$. Let $h_{\ell}(u)$ be the value of $h(u)$ at the point in time when $\hat{d}_{\tau}(u)$ was decremented to $\ell$. We note that if $h_{\ell}(u)\leq h^{NEW}(u)$ then $\ell$ is a multiple of $2^{h_{\ell}(u)}$. Thus, when $\hat{d}_{\tau}(u)$ was decremented to $\ell$, if $\hat{d}_{\tau}(v)>\hat{d}_{\tau}(u)+1$ then we added $(u,v)$ to $H$, which has the effect of decreasing $\hat{d}_{\tau}(v)$ to $\ell+1$. 
    Thus, once we finish processing the current edge update, we have $\hat{d}_{\tau}(v)-\ell\leq 1$. By definition, $\ell-\hat{d}_{\tau}(u)\leq 2^{h(u)}-1$, so we have $\hat{d}_{\tau}(v)-\hat{d}_{\tau}(u)\leq 2^{h(u)}$. 
    \item \uline{$v$ is added to $\mathcal{FN}(u)$:} Since we are assuming that $\hat{d}_{\tau}(v) > \hat{d}_{\tau}(u)$, the only way $v$ can be added to $\mathcal{FN}(u)$ is if the edge $(u,v)$ is inserted. In this case, if $\hat{d}_{\tau}(v)>\hat{d}_{\tau}(u)+1$, then the algorithm inserts $(u,v)$ into the set $H$, which has the eventual effect that $\hat{d}_{\tau}(v)\leq\hat{d}_{\tau}(u)+1$.
    %Again, a vertex $v$ can be added to $\mathcal{FN}(u)$ if either the edge $(u,v)$ is inserted, the distance $\hat{d}_{\tau}(u)$ decreases to a multiple of $2^{h(u)}$, or $h(u)$ increases. We already handled the second case. In the first case, when an edge $(u,v)$ is inserted and $\hat{d}_{\tau}(v)>\hat{d}_{\tau}(u)+1$, the algorithm inserts $(u,v)$ into the set $H$, which has the eventual effect that $\hat{d}_{\tau}(v)\leq\hat{d}_{\tau}(u)+1$. In the third case, if $h(u)$ is increased, but since each element that enters $\mathcal{FN}(u)$ as a result after $h(u)$ is increased is scanned, each of them is again added to $H$. 
\end{enumerate}
\end{proof}

%Since our algorithm is more involved, we also include a proof stating that the forward neighborhoods are indeed large.

Next, we prove a lower bound on the size of the forward neighborhoods.

\begin{lemma}\label{lem:sizefn}
For all $u\in V$, $|\mathcal{FN}(u)|\geq(2^{h(u)} - 1)  \frac{6 n\log n}{\eps\tau}$ at all times except lines~\ref{line:begin} to \ref{line:end} and during $\textsc{DecreaseHeaviness}(u)$.
\end{lemma}
\begin{proof}
The inequality in the lemma statement could be violated due to two events: 
\begin{itemize}
    \item \uline{$h(u)$ increases:} $\textsc{IncreaseHeaviness}(u)$ is the only procedure that can increase $h(u)$. $\textsc{IncreaseHeaviness}(u)$ specifically sets $h(u)$ so that it satisfies $|\mathcal{FN}(u)|\geq (2^{h(u)} - 1)  \frac{6 n\log n}{\eps\tau}$.
    \item \uline{$\mathcal{FN}(u)$ shrinks:} There are two scenarios that could cause $\mathcal{FN}(u)$ to shrink. Either, 1) $h(u)$ decreases, in which case it is set so that $|\mathcal{FN}(u)|\geq (2^{h(u)} - 1)  \frac{6 n\log n}{\eps\tau}$, or 2) a vertex $v\in\mathcal{FN}(u)$ has its distance estimate $\hat{d}_{\tau}(v)$ decremented causing $v$ to leave $\mathcal{FN}(u)$. In this case, $v$ leaves $\mathcal{FN}(u)$ only if $\hat{d}_{\tau}(v)$ decrements to $\texttt{CacheIndex}(u)-1$ and $v$'s position in $\texttt{Cache}_u$ is updated to $\texttt{Cache}_u[\hat{d}_{\tau}(v)]$. We observe that when $\hat{d}_{\tau}(v)$ decrements, we iterate through each vertex $w\in \normalfont{\texttt{Expire}}_v[\hat{d}_{\tau}(v)+1]$ (line~\ref{line:exp}). Since we update $\texttt{Expire}_v$ immediately every time $v$ is moved to a new index in $\texttt{Cache}_u$, we have that if $\hat{d}_{\tau}(v)=\texttt{CacheIndex}(u)-1$ then $u\in \normalfont{\texttt{Expire}}_v[\hat{d}_{\tau}(v)+1]$. Thus, if $v$ has left $\mathcal{FN}(u)$, then the loop on line~\ref{line:exp} calls $\textsc{DecreaseHeaviness}(u)$, which specifically sets $h(u)$ so that it satisfies $|\mathcal{FN}(u)|\geq (2^{h(u)} - 1)  \frac{6 n\log n}{\eps\tau}$.
\end{itemize}
\end{proof}

\noindent
We are now ready to prove the final lemma, establishing the correctness of the algorithm.

\begin{lemma}\label{lem:correct}
After processing each edge update, for each $t \in V$ and each $\tau$, $d(s,t)\leq\hat{d}_{\tau}(t)$ and if $d(s,t) \in [\tau, 2\tau)$ then $\hat{d}_{\tau}(t) \leq (1+\eps)d(s,t)$.
\end{lemma}
\begin{proof}
Our main argument is a generalization of the proof of correctness from the warm-up algorithm. Fix a heaviness level $h>0$. Let $s = t_0$. Then, we define $s_{i+1}$ be the first vertex with heaviness $h$ after $t_{i}$ on $\pi_{s,t}$ and let $t_{i+1}$ be the last vertex on $\pi_{s,t}$ of heaviness $h$ whose forward neighborhood intersects with the forward neighborhood of $s_{i+1}$ (possibly $t_{i+1} = s_{i+1}$). Thus, we get pairs $(s_1, t_1), (s_2, t_2), \dots, (s_k, t_k)$. Additionally, let $s_{k+1} = t$.

By definition, the forward neighborhoods of all $s_i$'s are disjoint. By Lemma \ref{lem:sizefn}, for each $s_i$, $|\mathcal{FN}(s_i)|\geq(2^h - 1)  \frac{6 n\log n}{\eps\tau}$ and since all $s_i$'s have disjoint forward neighborhoods, we have at most $k$ pairs $(s_i, t_i)$ with 
\[
k \leq \frac{n}{(2^h - 1)  \frac{6 n\log n}{\eps\tau}} \leq \frac{\eps\tau}{6 (2^h-1)\log n}.
\]

For any $i$, let $v_i$ be a vertex in $\mathcal{FN}(s_i) \cap \mathcal{FN}(t_i)$ (which exists by definition of $t_i$). By Invariant \ref{inv:correct}, we have $|\hat{d}_{\tau}(s_i) - \hat{d}_{\tau}(v_i)| \leq 2^{h}$ and $|\hat{d}_{\tau}(v_i) - \hat{d}_{\tau}(t_i)| \leq 2^{h}$. Thus, $\hat{d}_{\tau}(t_i) - \hat{d}_{\tau}(s_i) \leq 2^{h+1}$.

Let $t'_i$ be the vertex on $\pi_{s,t}$ succeeding $t_i$ (except $t'_0 = s$). If $t'_i\in \mathcal{FN}(t_i)$ then by Invariant \ref{inv:correct}, we have $\hat{d}_{\tau}(t'_i) - \hat{d}_{\tau}(t_i) \leq 2^{h}$ and otherwise, $t'_i\not\in \mathcal{FN}(t_i)$ so $\hat{d}_{\tau}(t'_i) < \texttt{CacheIndex}(t_i) < \hat{d}_{\tau}(t_i)$. So regardless, we have $\hat{d}_{\tau}(t'_i) - \hat{d}_{\tau}(t_i) \leq 2^{h}$. Combining this with the previous paragraph, we have $\hat{d}_{\tau}(t'_i) - \hat{d}_{\tau}(s_i) \leq 3*2^{h}$.

Now, let $h_{max} = \log n$ be the maximum heaviness level. We handle heaviness level $h'$ (initially $h_{max}$) by find the pairs $(s_i,t_i)$ %and corresponding $s_i$,$t_i$ and $t'_i$'s 
for heaviness $h'$ on the path $\pi'$ (initially $\pi_{s,t}$). This partitions the path $\pi'$ into segments $\pi'[t'_i, s_{i+1}]$ and $\pi'[s_{i+1}, t'_{i+1}]$. We observe that all arc tails in these path segments have heaviness less than $h'$. We contract the path segments $\pi'[s_{i+1}, t'_{i+1}]$ to obtain the new path $\pi'$, decrement $h'$ and recurse. 
%(Technically, $t'_i$ could also be heavy and therefore we need to contract the connected components in $\bigcup_i \pi'[s_{i+1}, t'_{i+1}]$ to handle empty light segments $\pi'[t'_i, s_{i+1}]$). 
We continue this scheme until $h'$ is $0$. By the previous analysis for each heaviness level $h'$, summing over the distance estimate difference of vertex endpoints of each contracted segment we obtain at most $\frac{3(2^{h'})\eps\tau }{6(2^{h'}-1)\log n}\leq \frac{\eps\tau}{\log n}$ (since $h'>0$) total error. Thus, each heaviness level larger than 0 contributes at most $\frac{\eps\tau}{\log n}$ additive error and overall they only induce additive error ${\eps\tau}$. 

For $h' = 0$, we argue that the algorithm induces no error on edges on $\pi'$ were each arc tail is of heaviness $0$. We will show that if $u$ is vertex of heaviness $0$ and $(u,v)$ is an edge, then $\hat{d}_{\tau}(v)\leq \hat{d}_{\tau}(u)+1$. This is straightforward to see from the algorithm description, but we describe the argument in detail for completeness. Consider the last of the following events that occurred: a) edge $(u,v)$ was inserted, b) $\hat{d}_{\tau}(u)$ was decremented, or c) the heaviness of $\hat{d}_{\tau}(u)$ became $0$. Case a occurs in the $\textsc{InsertEdge}(u,v)$ procedure where the algorithm decreases $\hat{d}_{\tau}(v)$ to be at most $\hat{d}_{\tau}(u)+1$. Case b occurs in the $\textsc{Decrement}(v)$ procedure. Here, the algorithm checks whether $\hat{d}_{\tau}(v)$ is a multiple of $2^{h(v)}$, which is true since $h(v)=0$. Then the algorithm updates the distance estimate of all vertices in $\mathcal{FN}(u)$, so if $\hat{d}_{\tau}(v)>\hat{d}_{\tau}(u)+1$ then $\hat{d}_{\tau}(v)$ is decreased to $\hat{d}_{\tau}(u)+1$. Case c occurs in the $\textsc{DecreaseHeaviness}(u)$ procedure where again the algorithm updates the distance estimate of all vertices in $\mathcal{FN}(u)$. 

By definition, the path $\pi'$ above is of length at most $d(s,t)$ and therefore we obtain an upper bound on $\hat{d}_{\tau}(t)$ of $d(s,t) + {\eps\tau}$. Then, when $d(s,t) \geq \tau$, the additive error of ${\eps\tau}$ is subsumed in the multiplicative $(1+\eps)$-approximation, as required.
\end{proof}

% We iterate by reducing the heaviness level by one and re-curse on the collection of path segments $\pi_{s,t}[t'_i, s_{i+1}]$ as described above. We note that if $s_i$ and $t_i$ are in different path segments, this only further reduces the worst-case additive error. On the final heaviness level $i = 2^0$ the difference between the starting and ending vertex of each segment is at most their distance in $G$. 

% We can therefore bound the total additive error by summing over the error induced on each level:
% \[
% \sum_{i = 1}^{\log n} k_i * 5 * 2^i \leq \sum_{i=1}^{\log n} \frac{\eps\tau}{\log n} = \eps \tau.
% \]
% Since each distance of interest is of size at least $\tau$, we can thus subsume the additive error in a $(1+\eps)$-multiplicative error. 

\subsection{Running time analysis}

We will show that the total running time of each data structure $\mathcal{E}_{\tau}$ is $\tilde{O}(n^2/\eps)$. Since there are $O(\log n)$ values of $\tau$, this implies that the total running time of the algorithm is $\tilde{O}(n^2/\eps)$. For the rest of this section we fix a value of $\tau$.

We crucially rely on the following invariant, which guarantees that  the heaviness of each vertex $u$ is chosen to be maximal, in the sense that if $h(u)$ were larger then we would have an upper bound on the size of $\mathcal{FN}(u)$.

\begin{invariant}\label{inv:time}
At all times, for all $u\in V$ and all integers $i$ such that $h(u)< i \leq \log n$,
    \[|\normalfont{\texttt{Cache}}_u[\texttt{CacheIndex}(u,2^i), \tau_{max}]| \leq  (2^i - 1) \frac{12 n \log n }{\eps\tau}.\]
\end{invariant}
\begin{proof}
We note the the invariant is satisfied on initialization since $\texttt{Cache}_u$ is initially empty. Let us now consider the events that could cause the invariant to be violated for some fixed $i$:
\begin{enumerate}
    \item \uline{$h(u)$ is decreased:} We note that $h(u)$ is only decreased in line \ref{line:sethd} of $\textsc{DecreaseHeaviness}(u)$, where it is set to a value that satisfies the invariant.
    (In particular, $h(u)$ cannot decrease in $\textsc{IncreaseHeaviness}(u)$ by Lemma~\ref{lem:sizefn}.)
   % In $\textsc{DecreaseHeaviness}(u)$, we set $h(u)$ in line \ref{line:sethd} to a value that satisfies the invariant. and afterwards do not change the indices of any items in $\texttt{Cache}_u$, thus when the procedure terminates the invariant remains true.
    \item \uline{A vertex $v$ is added to $\texttt{Cache}_u$:} This scenario could only occur due to an insertion of an edge $(u, v)$. However, after adding $v$ to $\texttt{Cache}_u$ (and $u$ to $\texttt{Expire}_v$), we directly invoke the procedure $\textsc{IncreaseHeaviness}(u)$, which we analyze below.
    \item \uline{$\texttt{CacheIndex}(u,2^i)$ is decreased:} Here, we note that $\texttt{CacheIndex}(u,2^i)$ decreases only if $\hat{d}_{\tau}(u)$ decreases to a multiple of $2^i$, in which case also call $\textsc{IncreaseHeaviness}(u)$.
\end{enumerate}

For the last two cases, it remains to prove that the procedure $\textsc{IncreaseHeaviness}(u)$ indeed resolves a violation of the invariant. If we do not enter the {\bf if} statement on line~\ref{line:ifloop1}, then by the definition of $i'$, the invariant is satisfied. If we do enter the {\bf if} statement, then invariant is satisfied for all $i>i'$. By Lemma \ref{lem:cachedecreases} the indices of vertices in $\texttt{Cache}_u$ can only decrease and therefore during the course of $\textsc{IncreaseHeaviness}(u)$, the size of $\normalfont{\texttt{Cache}}_u[\texttt{CacheIndex}(u,2^{i'}), \tau_{max}]$ can only decrease. Thus, when $\textsc{IncreaseHeaviness}(u)$ terminates, it is still the case that the invariant holds for all $i>i'$. On the other hand, if $i\leq i'$, then we set $h(u)$ on line \ref{line:calc2} so that the invariant is satisfied. 

%We then update the cache values, possibly decreasing the indices of items stored in $\texttt{Cache}_u$. We then select the new heaviness $h^{NEW}(u)$ in line \ref{line:calc2} to be the largest $i \leq i'$ with 
%\[
%|\normalfont{\texttt{Cache}}_u[\texttt{CacheIndex}(u,2^i), \tau_{max}]| \leq  (2^i - 1) \frac{6 n \log n }{\eps\tau}
%\]
%so certainly, no value $h(u), \dots, i'$ can violate our invariant. To see that no values larger than $i'$ can violate the invariant, we argue similarly, that $i'$ was initially chosen to be the maximal heaviness, thus for every $i'' > i'$ we had
%\[
%    |\normalfont{\texttt{Cache}}_u[\texttt{CacheIndex}(u,2^{i''}), \tau_{max}]| \leq  (2^{i''} - 1) \frac{12 n \log n }{\eps\tau}
%\]
%at the time that $i'$ was chosen. However, since the size of each $\normalfont{\texttt{Cache}}_u[\texttt{CacheIndex}(u,2^{i''}), \tau_{max}]$ can only decrease during the rest of the procedure since $\hat{d}_{\tau}(u)$ remains fixed, thus $i''$ satisfies the invariant when the procedure terminates.
\end{proof}

We can now prove the most important lemma of this section bounding the time spent in the loops starting at lines
\ref{line:fndecrement}, \ref{line:loop1}, \ref{line:loop2}, \ref{line:loop3} and \ref{line:loop4}.

\begin{lemma}\label{lem:iscan}
The total time spent in the loops starting in lines
\ref{line:fndecrement}, \ref{line:loop1}, \ref{line:loop2}, \ref{line:loop3} and \ref{line:loop4} is $O(n^2 \log^4 n / \eps)$.
\end{lemma}
\begin{proof}
We start our proof by pointing out that the time spent in the loop starting in line \ref{line:loop2} is subsumed by the time spent by the loop in line \ref{line:loop1} for the following reason. On line \ref{line:calc2} the heaviness is chosen so that the forward neighborhood is over a more narrow range of indices that in loop on line \ref{line:loop1}. Furthermore, By Lemma \ref{lem:cachedecreases} the indices of vertices in $\texttt{Cache}_u$ can only decrease and therefore between lines \ref{line:loop1} and \ref{line:loop2}, for all $i$ the size of $\normalfont{\texttt{Cache}}_u[\texttt{CacheIndex}(u,2^{i}), \tau_{max}]$ can only decrease.

Similarly, the running time spent in the loop starting in line \ref{line:loop4} is subsumed by the running time of the loop starting in line \ref{line:loop3}. Thus, we only need to bound the running times of the loops starting in lines \ref{line:fndecrement}, \ref{line:loop1}, and \ref{line:loop3}.

To bound their running times, we define the concept of $i$-scanning: we henceforth refer to the event of iterating through $\normalfont{\texttt{Cache}}_u[\texttt{CacheIndex}(u,2^{i}), \tau_{max}]$ by \emph{$i$-scanning $\normalfont{\texttt{Cache}_u}$}, for any $0 \leq i \leq \log n$, choosing the largest $i$ applicable.

Lines \ref{line:fndecrement}, \ref{line:loop1} and \ref{line:loop3} all correspond to $i$-scanning $\normalfont{\texttt{Cache}_u}$: the loop on line \ref{line:fndecrement} $h(u)$-scans $\normalfont{\texttt{Cache}_u}$, the loop on line \ref{line:loop1} $i'$-scans $\normalfont{\texttt{Cache}_u}$ for $i'$ chosen on line \ref{line:calc1}, and the loop at line \ref{line:loop3} $h(u)$-scans $\normalfont{\texttt{Cache}_u}$. We now want to bound the total number of $i$-scans in order to bound the total running time. 

\begin{claim}
\label{clm:iscans}
For all $u\in V$ and all integers $0\leq i\leq \log n $, the algorithm $i$-scans $\normalfont{\texttt{Cache}}_u$ at most $O(\tau \log^2 n /2^i)$ times over the course of the entire update sequence.
\end{claim}
\begin{proof}
We first observe that we $i$-scan $\texttt{Cache}_u$ on line \ref{line:fndecrement} only if we are in the procedure $\textsc{Decrement}(u', u)$ for some $u'$, and $\hat{d}_{\tau}$ is decreased to a value that is a multiple of $2^i$. Since each invocation of $\textsc{Decrement}(u', u)$, decreases $\hat{d}_{\tau}(u)$ by $1$ and since  $\hat{d}_{\tau}(u)$ is monotonically decreasing, starting at $\tau_{max} + 1$, we conclude that the number of $i$-scans on line~\ref{line:fndecrement} is bound by $O(\tau/2^i)$. 

Next, let us bound the number of $i$-scans executed in the loop starting on line \ref{line:loop1} in procedure $\textsc{IncreaseHeaviness}(u)$. We claim that between any two $i$-scans of $\texttt{Cache}_u$ on line \ref{line:loop1}, either $\hat{d}_{\tau}(u)$ becomes a multiple of $2^i$ or at least $(2^i-1) \frac{ n \log n }{\eps\tau}$ edges emanating from $u$ are inserted into the graph. Observe that this claim immediately implies that there can be at most $\tau/2^i + \frac{n}{2^i \frac{ n \log n }{\eps\tau}} = O(\tau \log n / 2^i)$ $i$-scans on line \ref{line:loop1}. 

To prove this claim, let $t_1$ and $t_2$ be two points in time at which $i$-scans occur. We will prove that if $\hat{d}_{\tau}(u)$ did not become a multiple of $2^i$ between times $t_1$ and $t_2$ then there were many edge insertions between times $t_1$ and $t_2$. Observe first, that $\texttt{CacheIndex}(u,2^i)$ only changes when $\hat{d}_{\tau}(u)$ decreases to become a multiple of $2^i$. Thus, we assume for the rest of the proof that $\texttt{CacheIndex}(u,2^i)$ remains fixed between times $t_1$ and $t_2$. Therefore, the size of $\texttt{Cache}_u[\texttt{CacheIndex}(u,2^i), \tau_{max}]$ can only be increased if a new edge $(u,v)$ is inserted with $v$ at distance $\hat{d}_{\tau}(v) \geq \texttt{CacheIndex}(u,2^i)$. 

Now, let $i'$ be such that at time $t_1$, we $i$-scan with $i' = i$ was selected in line \ref{line:calc1}. However, observe that since at $t_2$, we only $i'$-scan with $i' = i$, if $i' > h(u)$. Thus, at some point $t$ such that $t_1 \leq t < t_2$, we either decreased the heaviness to below $i'$ on line \ref{line:i'dec}, or we already set $h(u)$ to a smaller value than $i'$ at time $t_1$ in line \ref{line:calc2}. In either case we certified that
\[
    | \texttt{Cache}_u[\texttt{CacheIndex}(u,2^i), \tau_{max}]| < (2^i-1) \frac{6 n \log n }{\eps\tau}.
\]
Since again, at time $t_2$, we picked $i' = i$, we certified on line \ref{line:calc1} that, 
\[
    | \texttt{Cache}_u[\texttt{CacheIndex}(u,2^i), \tau_{max}]| \geq (2^i-1) \frac{12 n \log n }{\eps\tau}.
\]
We have shown that between times $t_1$ and $t_2$, the size of $\texttt{Cache}_u[\texttt{CacheIndex}(u,2^i), \tau_{max}]$ can only increase due to edge insertions. Thus, we conclude that at least $6(2^i-1) \frac{n \log n }{\eps\tau}$ edges with tail $u$ must have been inserted between times $t_1$ and $t_2$. 

Finally, we prove that the number of $i$-scans in the loop starting on line  \ref{line:loop3} is bounded. We first observe that each time an $h(u)$-scan is executed, we afterwards decrease the heaviness by at least one: By Lemma \ref{lem:cachedecreases} the indices of vertices in $\texttt{Cache}_u$ can only decrease and therefore between lines \ref{line:loop3} and \ref{line:sethd} for any $i$ the size of  $\normalfont{\texttt{Cache}}_u[\texttt{CacheIndex}(u,2^{i'}), \tau_{max}]$ can only decrease. Thus, when we pick the new $h(u)$, it satisfies $h(u) \leq i'$.

Now, we use the fact that there are at most $\log n$ heaviness values to bound the number of $i$-scans in the loop starting on line  \ref{line:loop3}. Since the number of vertices scanned when we increase $h(u)$ is more than the number of vertices scanned on line~\ref{line:loop3} when we decrease $h(u)$, the total number of vertices scanned in the loop on line~\ref{line:loop3} is at most $\log n$ times the number of vertices scanned in the loop on line~\ref{line:loop2}.
%we can use the upper bound on the number of increases to heaviness $i$ and larger (which is $\log n - i$ indices) and take it times the number of decreases. 
Thus, there are at most $O(\tau \log^2 n / 2^i)$ $i$-scans on line  \ref{line:loop3}.
\end{proof}

Now, the running time of each of these $i$-scans can be bound by $O(2^i \frac{n \log n}{\eps \tau})$ by Invariant \ref{inv:time}, so we obtain the claimed running time of
%Old version:
% Now, the running time of each of these $i$-scans can be bound by $O(2^i \frac{n \log n}{\eps \tau})$ since we either have for loops of lines \ref{line:fndecrement} and \ref{line:loop3} that Invariant \ref{inv:time} is enforced, therefore for each $h(u)$-scan, we have for $i = h(u)+1$, at most $O(2^i \frac{n \log n}{\eps \tau})$ items that we iterate over and for the loop in line \ref{line:loop1}, we choose $i'$ so that for $i=i'+1$ again there are at most $O(2^i \frac{n \log n}{\eps \tau})$ items in $\texttt{Cache}_u[\texttt{CacheIndex}(u,2^i), \tau_{max}]$ which contains each scanned item. It is straight-forward to see that each operation in the loops takes only $O(1)$. Using again the fact that there are $\log n$ values of $i$, we obtain the claimed running time of
\[
    \sum_i O\left((\tau \log^2 n /2^i)\left(2^i \frac{n \log n}{\eps \tau}\right)\right)= O(n \log^4 n / \eps).
\]
\end{proof}

We can now reuse claim \ref{clm:iscans} to bound the total time spent in the loop on line \ref{line:while} in the procedure $\textsc{InsertEdge}(u,v)$.

\begin{lemma}\label{lem:insertEdge}
The total running time spent in the loop starting on line \ref{line:while} excluding calls to $\textsc{Decrement}(u,v)$ is bounded by $O(n^2 \log^4 n / \eps)$.
\end{lemma}
\begin{proof}
We first observe that on line \ref{line:Hinsert1} we only add newly inserted edges into $H$. Thus, we add a total of at most $n^2$ edges to $H$ during line \ref{line:Hinsert1}. The remaining edges are only inserted into $H$ during $i$-scans in the lines \ref{line:Hinsert2} and \ref{line:Hdec}. Since by claim \ref{clm:iscans} there are at most $O(\tau \log^2 n/2^i)$ $i$-scans of $\texttt{Cache}_u$ for any $u\in V$, and each $i$-scan is over at most $O(2^in \log n / \eps \tau)$ elements, similarly to the preceding lemma, we conclude that we iterate over at most $O(n \log^4 n/ \eps)$ elements in all $i$-scans of $\texttt{Cache}_u$ over all values of $i$, for a fixed $u\in V$. Since each element that we iterate over in each $i$-scan can only result in the insertion of a single edge into $H$, we can bound the total number of insertions into $H$ over the entire course of the algorithm by $O(n^2 \log^4 n / \eps)$. Further, we observe that each iteration of the loop in line \ref{line:while} either removes an edge from the set $H$, or decrements a distance estimate, we can bound the total number of iterations of the loop by $O(n^2 \log^4 n / \eps) + n \tau_{max} = O(n^2 \log^4 n / \eps)$. Since each iteration takes $O(1)$ time, ignoring calls to $\textsc{Decrement}(u,v)$, the lemma follows.
\end{proof}

We are now ready to finish the running time analysis.

\begin{lemma} \label{lma:mainResultUnweightedUpperBound}
The total running time of a data structure $\mathcal{E}_{\tau}$ is $O(n^2 \log^5 n / \eps)$.
\end{lemma}
\begin{proof}
We begin with the procedure $\textsc{InsertEdge}(u,v)$. We note that this procedure takes constant time except for the {\bf while} loop, if we ignore the calls to $\textsc{IncreaseHeaviness}(u)$. Since there are at most $n^2$ edge insertions, the running time can be bounded by $O(n^2)$. Further, the total running time spend in the {\bf while} loop starting in line \ref{line:while} excluding calls to $\textsc{Decrement}(u,v)$ is bounded by $O(n^2 \log^4 n / \eps)$ by lemma \ref{lem:insertEdge}.

Next, let us bound the total time spent in procedure $\textsc{Decrement}(u,v)$. We first observe that the loop on line \ref{line:exp} iterates through each vertex $w$ in $\texttt{Expire}_u[\hat{d}_{\tau}(u)+1]$ removing each $w$ from $\texttt{Expire}_v$. Clearly, the number of iterations over the course of the entire algorithm can be bounded by the total number of times a vertex is inserted into $\texttt{Expire}_v$ over all $v$. Since these insertions occur in the loops starting in lines \ref{line:loop2} and \ref{line:loop4}, we have by lemma \ref{lem:iscan}, that the time spend on the loop starting in line \ref{line:exp} is bound by $O(n^2 \log^4 n/\eps)$. Further, ignoring subcalls, each remaining operation in the procedure $\textsc{Decrement}(u,v)$ takes constant time. We further observe that since each invocation of the procedure $\textsc{Decrement}(u,v)$ decreases a distance estimate, the procedure is invoked at most $n\tau_{max} = O(n^2)$ times. Thus, we can bound the total time spent in procedure $\textsc{Decrement}(u,v)$ by $O(n^2 \log^4 n/\eps)$.

For the remaining procedures $\textsc{IncreaseHeaviness}(u)$ and $\textsc{DecreaseHeaviness}(u)$, we note that the calculations of $i'$ and $h(u)$ on lines~\ref{line:calc1}, \ref{line:calc2}, \ref{line:i'dec}, and \ref{line:sethd} can be implemented in $O(\log n)$ time using a binary tree over the elements of array $\texttt{Cache}_u$ for each $u \in V$. %Thus, excluding the loops, The remaining part of the procedures excluding the loops can be implemented in $O(\log n)$ time. 
We observe that both procedures receive at most $O(n^2 \log^4 n / \eps)$ invocations and since we already bounded the running times of the loops that call them. Thus, the total update time excluding loops can be bound by $O(n^2 \log^5 n/\eps)$. The loops take total time $O(n^2 \log^4 n/\eps)$ by Lemma \ref{lem:iscan}. % and \ref{lem:insertEdge}. %(observe that due to the binary tree structure over $\texttt{Cache}_u$ the update time is increased by an $O(\log n)$ factor). 
This concludes the proof.
\end{proof}

Using $\log n$ data structures, one for each distance threshold $\tau$, we obtain the following result.

\begin{theorem} \label{thm:mainResultUnweightedUpperBound}
There is a deterministic algorithm that given an unweighted directed graph $G=(V,E)$, subject to edge insertions, a vertex $s\in V$, and $\eps>0$, maintains for every vertex $v$ an estimate $\hat{d}(v)$ such that after every update $d(s,v) \leq \hat{d}(v) \leq (1+\eps)d(s,v)$, and runs in total time $O(n^2 \log^5 n/\eps)$. A query for the approximate shortest path from $s$ to any vertex $v$ can be answered in time linear in the number of edges on the path.
\end{theorem}

\subsubsection{Weighted graphs}

Finally, we show how to extend our data structure to deal with weights $[1, W]$. We first show how to handle edge weights with a linear dependency in the running time on $W$. Then, we employ a standard edge-rounding technique \cite{raghavan1987randomized, cohen1998fast, zwick2002all, bernstein2009fully, madry2010faster, bernstein2016maintaining} that decreases the dependency in $W$ to $\log W$ (we will use a set-up most similar to \cite{bernstein2016maintaining}).

\begin{lemma} \label{lma:mainResultweightedUpperBound}
There is a deterministic algorithm that given a weighted directed graph $G=(V,E, w)$, subject to edge insertions and weight changes, with weights in $[1, W]$, a vertex $s\in V$, and $\eps>0$, maintains for every vertex $v$ an estimate $\hat{d}(v)$ such that after every update $d(s,v) \leq \hat{d}(v) \leq (1+\eps)d(s,v)$ if $d(s,v) \in [\tau, 2\tau)$ for some $\tau \leq n$, and runs in total time $O(n^2 \log^5 n/\eps^{1.5})$. A query for the approximate shortest path from $s$ to any vertex $v$ can be answered in time linear in the number of edges on the path.
\end{lemma}
\begin{proof}
Let us first describe an almost correct approach to modify the data structure $\mathcal{E}_{\tau}$ for unweighted graphs to handle edge weights and maintains shortest-paths of weight at most $\tau_{max}$ as follows: we change the if-condition in line  \ref{line:ifThenhLoop} from $\hat{d}_{\tau}(v)>\hat{d}_{\tau}(u)+1$ to $\hat{d}_{\tau}(v)>\hat{d}_{\tau}(u)+ w(u,v)$ and similarly in line \ref{line:ifThenDecrement} to $\hat{d}_{\tau}(y)>\hat{d}_{\tau}(x)+w(x,y)$. Further, we need to adapt indices in $\texttt{Cache}_u$ and $\texttt{Expire}_u$ accordingly to reflect the additional offset which is straightforward. 

Unfortunately, whilst the running time can still be bound as before, the correctness of the algorithm could no longer be guaranteed since invariant \ref{inv:correct} is no longer true. Recall that the invariant states that if $v\in \mathcal{FN}(u)$ then $|\hat{d}_{\tau}(v) - \hat{d}_{\tau}(u)| \leq 2^{h(u)}$. However, a vertex $u$ might now have a vertex $v$ in its forward-neighborhood at large distance but have a large edge weight on $(u,v)$ so it can not decrease its distance estimate.

However, a rather simple fix suffices: whenever we compute the heaviness $i$ by setting it to 
\[
\argmax_{i \in \mathbb{N}}\{| \texttt{Cache}_u[\texttt{CacheIndex}(u,2^i), \tau_{max}]| \geq (2^i-1) \frac{6 n \log n }{\eps\tau} \}
\]
we now no longer want to take all vertices in $\texttt{Cache}_u[\texttt{CacheIndex}(u,2^i), \tau_{max}]$ into account but only all neighbors $v$ such that the edge $(u,v)$ is of edge weight less than $2^i$ (observe that heaviness levels now depend on different sets). Similarly, we use the restriction on the neighbors for reducing heaviness, and it is only these edges that we then consider to be in the forward neighborhood. It is straightforward to conclude that invariant \ref{inv:correct} can be restored to guarantee that $v\in \mathcal{FN}(u)$ implies $|\hat{d}_{\tau}(v) - \hat{d}_{\tau}(u)| \leq 2 * 2^{h(u)}$. 

However, this change alone is not enough to get good running time. We also stipulate that each edge $(u,v)$ is scanned only every $\eps w(u,v)$ levels if $v \not\in \mathcal{FN}(u)$. It is straightforward to verify that this might induce a multiplicative error of $(1+\eps)$ on every edge. However, by rescaling $\eps$ by a constant factor, we can still conclude that by the restored invariant \ref{inv:correct}, the proof \ref{lem:correct} works as before and guarantees a $(1+\eps)$ multiplicative error on distances in $[\tau, 2\tau)$. 

Now let us bound the running time where we only bound the running time induced by scanning the weighted edges as described above since the bounds on the remaining running time carry seamlessly over from lemma \ref{lma:mainResultUnweightedUpperBound}. It can be verified that invariant \ref{inv:time} is still enforced for our new definition. Thus, if the heaviness is $h(u) = i$ for some vertex $u$, then the number of edges of weight in $(2^j, 2^{j+1}]$ for $j > i$ is at most $(2^j-1) \frac{12 n \log n }{\eps\tau}$. Since we scan these edges only every $\eps 2^j$ decrements of $\hat{d}_{\tau}(u)$, we obtain that the total running time required for all edge scans can be bound by
\[
    \sum_{v \in V} \sum_{j \in (0, \log n]} O\left( \left(2^j \frac{ n \log n }{\eps\tau} \right) \left(\frac{\tau_{max}}{\eps2^j} \right) \right) = O(n^2 \log^2 n /\eps^2).
\]
We point out that rebalancing terms slightly, we can reduce the $\eps$ dependency to $1/\eps^{1.5}$.
\end{proof}

We now prove the following lemma which implies theorem \ref{thm:mainResultUpperBound} as a corollary by maintaining a data structure $\mathcal{E}_{\tau_{hop}, \tau_{depth}}$ with parameters $\tau_{hop} = 2^i$ and $\tau_{depth}=2^j$, for every $i \in [0, \log n)$ and $j \in [0, \log nW)$. We point out that we define \emph{length} subsequently as the number of edges on a path and \emph{weight} as the sum over all edge weights on a path.

\begin{lemma}
There is a deterministic data structure $\mathcal{E}_{\tau_{hop}, \tau_{depth}}$ that given a weighted directed graph $G=(V,E, w)$, subject to edge insertions and weight changes, with weights in $[1, W]$, that takes parameters $\tau_{hop}$ and $\tau_{depth} \geq \tau_{hop}$, a vertex $s\in V$, and $\eps>0$, and maintains for every vertex $v$ with some shortest path in $G$ consisting of $[\tau_{hop}, 2\tau_{hop})$ edges and of weight in $[\tau_{depth}, 2\tau_{depth})$, an estimate $\hat{d}(v)$ such that after every update $d(s,v) \leq \hat{d}(v) \leq (1+\eps)d(s,v)$ and runs in total time $O(n^2 \log^8 n/\eps^{2.5})$. A query for the approximate shortest path from $s$ to any vertex $v$ can be answered in time linear in the number of edges on the path.
\end{lemma}
\begin{proof}
Let us start by defining some constant $\alpha = \frac{ \eps \tau_{depth} }{ \tau_{hop} }$ (we assume that $\alpha$ is integer by slightly perturbing $\eps$). Then, we let $G_{\alpha}$ be the graph $G$ after rounding each edge up to the nearest multiple of $\alpha$. We claim that for every vertex $t \in V$, for which we have a shortest path $\pi_{s,t}$ from $s$ to $t$ of length in $[\tau_{hop}, 2\tau_{hop})$ and weight in $[\tau_{depth}, 2\tau_{depth}]$, we have
\[
    w_{G_{\alpha}}(\pi_{s,t}) \leq w_G(\pi_{s,t}) \leq (1+2\eps)w_{G_{\alpha}}(\pi_{s,t}).
\]
To see this observe that each edge incurs additive error at most $\alpha$. However, since the path is of length at most $2 \tau_{hop}$, the additive error has to be bound by $2\alpha \tau_{hop} =  2\frac{ \eps \tau_{depth} }{ \tau_{hop} } \tau_{hop} = 2 \eps \tau_{depth}$. But since the path $\pi_{s,t}$ is of weight at least $\tau_{depth}$, we have overall at most a $(1+3\eps)$-approximation and therefore by rescaling $\eps$ by a constant factor, the claim follows.

Next, we let $G^*_{\alpha}$ be the graph $G_{\alpha}$ where each edge is scaled down by factor $\alpha$ and note that weights are all integral and positive. We next claim that for every vertex $t \in V$, for which we have a shortest path $\pi_{s,t}$ from $s$ to $t$ of length in $[\tau_{hop}, 2\tau_{hop})$ and weight in $[\tau_{depth}, 2\tau_{depth}]$, we have
\[
    w_{G^*_{\alpha}}(\pi_{s,t}) \leq \tau_{hop}/\eps
\]
To see this, observe that the path $\pi_{s,t}$ in $G_{\alpha}$ has weight at most $(1+2\eps)2\tau_{depth}$ by our preceding claim. Thus, scaling it down by $\alpha$, the path has weight at most
\[
    (1+2\eps)2\tau_{depth} / \alpha = (1+2\eps)2\tau_{depth} \frac{ \tau_{hop} }{ \eps \tau_{depth} } = (1+2\eps)2 \tau_{hop} / \eps \leq 8 \tau_{hop} / \eps.
\]
in $G^*_{\alpha}$. It now remains to run a data structure $\mathcal{E}_{\tau}$ on $G^*_{\alpha}$ with $\tau = \tau_{hop}$ as described in Theorem \ref{lma:mainResultweightedUpperBound}, however run to depth $8\tau_{hop}/ \eps$ (instead of $\tau_{max}$ which increases the running time by an $1/\eps$ factor. We then forward for each vertex $t$, the distance estimate $\hat{d}_{\tau}(t)$ scaled up by $\alpha$. This concludes the lemma.
\end{proof}

%%%start LOWER BOUNDS
\input{lowerbounds}
\printbibliography[heading=bibintoc] % Make bibliography show up in table of contents
\appendix
\input{app}

\end{document}

%% file: commands.tex
\allowdisplaybreaks[1]

\makeatletter
\g@addto@macro\bfseries{\boldmath}
\makeatother

\makeatletter
\g@addto@macro\mdseries{\unboldmath}
\g@addto@macro\normalfont{\unboldmath}
\g@addto@macro\rmfamily{\unboldmath}
\g@addto@macro\upshape{\unboldmath}
\makeatother

\DeclareCiteCommand{\citem}
    {}
    {\mkbibbrackets{\bibhyperref{\usebibmacro{postnote}}}}
    {\multicitedelim}
    {}

\renewcommand*{\multicitedelim}{\addcomma\space}

\AtEveryCitekey{\clearfield{note}}

\newcommand{\myhref}[1]{%
  \iffieldundef{doi}
    {\iffieldundef{url}
       {#1}
       {\href{\strfield{url}}{#1}}}
    {\href{http://dx.doi.org/\strfield{doi}}{#1}}%
}

\DeclareFieldFormat{title}{\myhref{\mkbibemph{#1}}}
\DeclareFieldFormat
  [article,inbook,incollection,inproceedings,patent,thesis,unpublished]
  {title}{\myhref{\mkbibquote{#1\isdot}}}

\makeatletter
\AtBeginDocument{%
    \newlength{\temp@x}%
    \newlength{\temp@y}%
    \newlength{\temp@w}%
    \newlength{\temp@h}%
    \def\my@coords#1#2#3#4{%
      \setlength{\temp@x}{#1}%
      \setlength{\temp@y}{#2}%
      \setlength{\temp@w}{#3}%
      \setlength{\temp@h}{#4}%
      \adjustlengths{}%
      \my@pdfliteral{\strip@pt\temp@x\space\strip@pt\temp@y\space\strip@pt\temp@w\space\strip@pt\temp@h\space re}}%
    \ifpdf
      \typeout{In PDF mode}%
      \def\my@pdfliteral#1{\pdfliteral page{#1}}% I don't know why % this command...
      \def\adjustlengths{}%
    \fi
    \ifxetex
      \def\my@pdfliteral #1{}% isn't equivalent to this one
      \def\adjustlengths{\setlength{\temp@h}{-\temp@h}\addtolength{\temp@y}{1in}\addtolength{\temp@x}{-1in}}%
    \fi%
    \def\Hy@colorlink#1{%
      \begingroup
        \ifHy@ocgcolorlinks
          \def\Hy@ocgcolor{#1}%
          \my@pdfliteral{q}%
          \my@pdfliteral{7 Tr}% Set text mode to clipping-only
        \else
          \HyColor@UseColor#1%
        \fi
    }%
    \def\Hy@endcolorlink{%
      \ifHy@ocgcolorlinks%
        \my@pdfliteral{/OC/OCPrint BDC}%
        \my@coords{0pt}{0pt}{\pdfpagewidth}{\pdfpageheight}%
        \my@pdfliteral{F}% Fill clipping path (the url's text) with
        \my@pdfliteral{EMC/OC/OCView BDC}%
        \begingroup%
          \expandafter\HyColor@UseColor\Hy@ocgcolor%
          \my@coords{0pt}{0pt}{\pdfpagewidth}{\pdfpageheight}%
          \my@pdfliteral{F}% Fill clipping path (the url's text)
        \endgroup%
        \my@pdfliteral{EMC}%
        \my@pdfliteral{0 Tr}% Reset text to normal mode
        \my@pdfliteral{Q}%
      \fi
      \endgroup
    }%
}
\makeatother

%% file: intro.tex
A dynamic graph $G$ is a sequence of graphs $G_0, G_1, \dots , G_t$ such that $G_0$ is the \emph{initial} graph that is subsequently undergoing \emph{edge updates} such that every two consecutive versions $G_i$ and $G_{i+1}$ of the dynamic graph $G$ differ only in one edge (or the weight of one edge). If the sequence of update operations consists only of edge deletions and weight increases, we say that $G$ is a \emph{decremental} graph and if the update operations are restricted to edge insertions and weight decreases, we say that G is \emph{incremental}. In either case, we say that $G$ is \emph{partially dynamic} and if the update sequence is mixed we say $G$ is \emph{fully dynamic}. 

In the study of dynamic graph algorithms, we are concerned with maintaining properties of $G$ efficiently. More precisely, we are concerned with designing a data structure that supports \emph{update} and \emph{query} operations such that after the $i^{th}$ edge update is processed, an adversary can query properties of $G_i$. 

We consider the problem of (approximate) Single-Source Shortest Paths (SSSP) in a partially dynamic graph $G$. In this problem, a dedicated source vertex $s\in V$ is given on initialization and the query operation takes as input any vertex $t\in V$ and outputs the (approximate) shortest-path distance estimate $\hat{d}(s,t)$ from $s$ to $t$ in the current version of $G$. We say that distance estimates have \emph{stretch} $\alpha \geq 1$, if the algorithm guarantees that $d(s,t) \leq \hat{d}(s,t) \leq \alpha d(s,t)$ is satisfied for every distance estimate where $d(s,t)$ denotes the distance from $s$ to $t$ in the current version of $G$. %Additionally, the adversary can query for a corresponding shortest-path. 
%I commented out the above because if we put this in the defn of the problem, readers might assume that past results allow this type of query

When proving lower bounds for partially dynamic SSSP, we also consider a potentially easier problem, $s$-$t$ Shortest Path ($s$-$t$ SP), thus obtaining stronger lower bounds. In $s$-$t$ SP, one wants to maintain a shortest path from $s$ to $t$ for some fixed $s$ and $t$.
%that is, unlike in SSSP, here both $s$ and $t$ are fixed and the queries are only about the distance between them. 

\subsection{Motivation}

Partially dynamic SSSP is a well-motivated problem with wide-ranging applications:
\begin{itemize}
    \item Partially dynamic data structures are often used as internal data structures to solve the fully dynamic version of the problem (see for example \cite{king1999fully, roditty2004dynamic, henzinger2016dynamic} for applications of partially dynamic SSSP) which in turn can be used to maintain properties of real-world graphs undergoing changes.
    \item Partially dynamic SSSP is often employed as internal data structure for related problems such as maintaining the diameter in partially dynamic graphs \cite{ancona2018algorithms, DBLP:journals/corr/abs-1812-01602} or matchings in incremental bipartite graphs \cite{bernstein2018online}.
    \item Many static algorithms use partially dynamic algorithms as a subroutine. For example, incremental All-Pairs Shortest Paths can be used to construct light spanners \cite{alstrup2017constructing} and greedy spanners. Moreover, a recent line of research shows that many flow problems can be reduced to decremental SSSP, and recent progress has already led to faster algorithms for multi-commodity flow \cite{madry2010faster}, vertex-capacitated flow, and sparsest vertex-cut \cite{Chuzhoy:2019:NAD:3313276.3316320}.
\end{itemize}

\subsection{Prior Work}

In this section we discuss prior work directly related to our results. We use $\tilde{O}$ notation to suppress factors of $\log n$. 
We refer the reader to Appendix~\ref{app} for further discussion of related work.

Let $m$ be the maximum number of edges and let $n$ be the maximum number of vertices\footnote{Some algorithms allow vertex updates, therefore the number of vertices might be due to change.} in any version of the dynamic input graph $G$. If $G$ is weighted, we denote by $W$ the aspect ratio of the graph, which is the largest weight divided by the smallest weight in the graph. For partially dynamic algorithms, we follow the convention of stating the \emph{total update time} rather than the time for each individual update. Unless otherwise stated, queries take worst-case constant time.

\paragraph{Algorithms for partially dynamic directed SSSP.} For directed graphs, the classic ES-tree data structure by Even and Shiloach \cite{shiloach1981line} and its later extensions by Henzinger and King \cite{henzinger1995fully} initiated the field, with total update time $O(mn W)$ for exact incremental/decremental directed SSSP. Using an edge rounding technique \cite{raghavan1987randomized, cohen1998fast, zwick2002all, bernstein2009fully, madry2010faster, bernstein2016maintaining}, the ES-tree can further handle edge weights more efficiently, giving an $\tilde{O}(mn \log W / \eps)$ time algorithm for incremental/decremental $(1+\eps)$-approximate directed SSSP. This result has been improved to total update time $\tilde{O}(\min\{m^{7/6} n^{2/3}\log W, m^{2/3} n^{4/3 +o(1)}\log W\})  = mn^{9/10+o(1)} \log W$ by the breakthrough results of Henzinger, Forster, and Nanongkai \cite{henzinger2014sublinear,henzinger2015improved}. Their algorithm is Monte Carlo and works against an oblivious adversary (an adversary that fixes the entire graph sequence of updates in advance). Whilst presented only in the decremental setting, this algorithm appears to extend to the incremental setting.

%we believe that it might be possible to extend this data structure to incremental graphs.  

Very recently, Probst and Wulff-Nilsen \cite{directedDecrementalSSSP} improved upon this result and presented a randomized data structure for decremental directed SSSP against an oblivious adversary with total update time $\tilde{O}(\min\{mn^{3/4} \log W, m^{3/4} n^{5/4} \log W\})$. They also give a Las Vegas algorithm with  total update time $\tilde{O}(m^{3/4}n^{5/4} \log W)$  that works against an adaptive adversary.
%that is an adversary that can base the update sequence on the query output of the algorithm. 
They also get slightly improved bounds for unweighted decremental graphs. We point out, however, that their data structure cannot return approximate shortest paths to the adversary as it would reveal the random choices. Further, unlike the data structure by Henzinger et al., %\cite{henzinger2014sublinear,henzinger2015improved}
their approach cannot be extended to the incremental setting since it relies heavily on finding efficient separators and on maintaining the topological order of vertices in the graph. 

In summary, all known partially dynamic algorithms for directed graphs that are faster than ES-trees are {\em randomized}, and their \emph{amortized} update time even for $m \sim n^2$ insertions is at least some polynomial. Moreover, it is unclear how to extend the result from \cite{directedDecrementalSSSP} to the incremental setting.

\paragraph{Lower bounds.} Conditional lower bounds for partially dynamic SSSP were first studied by Roditty and Zwick \cite{roditty2004dynamic}. They showed that in the weighted setting, APSP can be reduced to partially dynamic SSSP with $O(n^2)$ updates and queries, thus implying that the amortized query/update time must be $n^{1-o(1)}$, unless APSP can be solved in truly subcubic time (i.e. $n^{3-\eps}$ for constant $\eps>0$).

For unweighted SSSP in the partially dynamic setting, there is a weaker lower bound \cite{roditty2004dynamic}: Under the Boolean Matrix Multiplication (BMM) hypothesis, any combinatorial incremental/decremental algorithm for unweighted SSSP requires amortized $n^{1-o(1)}$ update/query time.
Abboud and Vassilevska Williams~\cite{popularconj} modified the \cite{roditty2004dynamic} construction to give stronger lower bounds even for the unweighted $s$-$t$ SP problem: any combinatorial incremental/decremental algorithm for unweighted $s$-$t$ SP requires either amortized $n^{1-o(1)}$ update time or $n^{2-o(1)}$ query time.

We point out that the unweighted SSSP lower bounds of \cite{popularconj,roditty2004dynamic} are weak in two ways: (1) they are only for combinatorial algorithms, and (2) they hold only when the number of edges $m$ is quadratic in the number of vertices, so in terms of $m$, the update lower bound is merely $m^{0.5-o(1)}$. Henzinger et al.~\cite{henzinger2015unifying} aimed to rectify (1). They introduced a very believable assumption, the OMv Hypothesis, which is believed to hold for arbitrary algorithms, not merely combinatorial ones. Henzinger et al.~\cite{henzinger2015unifying} showed that under the OMv Hypothesis, incremental/decremental $s$-$t$ SP (in the word-RAM model) requires $m^{0.5-o(1)}$ amortized update time\footnote{Similar to the lower bounds based on BMM, the Henzinger et al.~\cite{henzinger2015unifying} lower bound on the update time is $n^{1-o(1})$ but the number of edges in the construction is quadratic in $n$, so that in terms of $m$, the lower bound is $m^{0.5-o(1)}$.} or $m^{1-o(1)}$ query time, thus obtaining the same lower bounds as under the BMM Hypothesis, but now for not necessarily combinatorial algorithms.

\subsection{Results}

Our main result is a new elegant algorithm for the incremental SSSP problem in weighted digraphs.

\begin{theorem} \label{thm:mainResultUpperBound}
There is a deterministic algorithm that given a weighted directed graph $G=(V,E)$ subject to $\Delta$ edge insertions and weight decreases, a vertex $s\in V$, and $\eps>0$, maintains for every vertex $v$ an estimate $\hat{d}(v)$ such that after every update $d(s,v)\leq \hat{d}(v)\leq (1+\eps)d(s,v)$, and runs in total time $\tilde{O}(n^2 \log W/\eps^{2.5} + \Delta)$. A query for the approximate shortest path from $s$ to any vertex $v$ can be answered in time linear in the number of edges on the path.
\end{theorem}

Our result is the first deterministic partially dynamic directed SSSP algorithm to improve over the long-standing $O(mn)$ time bound achieved by the ES-tree \cite{shiloach1981line}. Our result is essentially optimal for very dense graphs, and is the first algorithm with essentially optimal update time for any density in directed graphs. Furthermore, our algorithm further improves on the \emph{randomized} $mn^{0.9+o(1)} \log W$ time algorithm of Henzinger et al.~\cite{henzinger2014sublinear} if $m=\omega(n^{1.1})$ (their paper presents only in the decremental setting, but it appears to extend to the incremental setting as well).
%, which we believe can be extended to the incremental setting. 

A further strength of our algorithm is that in addition to returning distance estimates, it can also return the corresponding approximate shortest paths, i.e. it is path-reporting. All known path-reporting dynamic SSSP algorithms except for the ES-tree are randomized against an oblivious adversary, so our algorithm is the first path-reporting deterministic or randomized against an adaptive adversary algorithm even considering algorithms for undirected graphs. 
%Besides the ES-tree, our algorithm is the first 
%can report paths in a directed graph and that works against an adaptive adversary and is the 
%deterministic path-reporting data structure even for undirected graphs 
(A recent randomized data structure by Chuzhoy and Khanna \cite{Chuzhoy:2019:NAD:3313276.3316320} can return paths in undirected graphs and works against an adaptive adversary however it works in a more restricted setting and requires $n^{1+o(1)}$ query time). 

Finally, we point out that our theoretical bounds are also likely to translate into an algorithm that is fast in practice since we only rely on array and simple arithmetic operations and do not make use of any involved internal data structures, whilst previous partially dynamic SSSP algorithms are rather involved and often rely on complicated techniques. 

Our second contribution includes several new fine-grained lower bounds for the partially dynamic SSSP and $s$-$t$-SP problems in unweighted undirected graphs.
The only known conditional lower bounds for partially dynamic SSSP and $s$-$t$-SP in unweighted graphs give an update time lower bound of $m^{0.5-o(1)}$. While the ES-tree data structure does achieve an $O(\sqrt m)$ amortized update/query time upper bound whenever $m=\Theta(n^2)$, this upper bound does not improve for lower sparsities. This motivates the following question:

$ $\\
{\em Is partially dynamic SSSP solvable with amortized update/query time $O(\sqrt m)$ for all sparsities $m$?}
$ $\\

Our work answers this question with the tools of fine-grained complexity.
Our first result is based on the following $k$-Cycle hypothesis (see \cite{lincolnsoda18,ancona2019}).

% The \cite{rzesa04} lower bound was improved by Abboud and Vassilevka Williams \cite{focs14} showing that if Triangle detection in $n$ node graphs requires $n^{c-o(1)}$ time for some $c>2$, then incremental and decremental SSSP require either $n^{c-o(1)}$ preprocessing time, or amortized $n^{c-2-o(1)}$ update time or query time or $n^{c-1-o(1)}$ query time, boosting the query time lower bound to $n^{1.373-o(1)}$ for the current value of $c$. 

%on the $k$-Cycle rst, byHypothesis.

\begin{hypothesis}[$k$-Cycle Hypothesis] \label{hyp:cycle}
In the word-RAM model with $O(\log m)$ bit words,
for any constant $\eps > 0$, there exists a constant integer $k$, so that there is no $O(m^{2-\eps})$ time algorithm that can detect a $k$-cycle in an $m$-edge graph.
\end{hypothesis}

Our first result says that under the $k$-Cycle Hypothesis, if the preprocessing time of a partially dynamic $s$-$t$-SP  algorithm is subquadratic $O(m^{2-\eps})$ for $\eps>0$, then in fact the algorithm cannot achieve truly sublinear, $O(m^{1-\eps'})$ amortized update and query time for any $\eps'>0$. This is a quadratic improvement over the previous known lower bounds, and it is also tight, as trivial recomputation 
%and ES-trees 
achieves amortized update/query time $O(m)$.

\begin{theorem}
\label{thm:MainCycleLowerBound}
Under Hypothesis \ref{hyp:cycle}, there can be no constant $\eps > 0$ such that partially dynamic $s$-$t$ SP in undirected graphs can be solved with $O(m^{2-\eps})$ preprocessing time, and $O(m^{1-\eps})$ update \emph{and} query time, for all graph sparsities $m$.
\end{theorem}

A consequence of the proof of Theorem~\ref{thm:MainCycleLowerBound} above is that (under Hypothesis  \ref{hyp:cycle}) the $O(mn)$ total update time achieved by ES-trees is essentially optimal, also when $m$ is close to linear in $n$. Recall that the OMv lower bound only showed this for $m =\Theta(n^2)$.

While the above lower bound is tight, it only holds for truly subquadratic preprocessing time. Recall that the only known lower bound for arbitrary polynomial preprocessing time is the $m^{0.5-o(1)}$ bound under OMv. 

We first develop an intricate reduction that shows that an efficient enough partially dynamic $s$-$t$ SP algorithm can be used to solve the $4$-Clique problem. Then we define an {\em online} version of $4$-Clique, similar to OMv that is plausibly hard even for arbitrary polynomial time preprocessing. 

We show that if $4$-Clique requires $n^{c-o(1)}$ time for some $c$, then, then any algorithm for partially dynamic $s$-$t$ SP with $O(m^{c/2-\eps})$ preprocessing time for some $\eps>0$, must have update or query time at least $m^{(c-2)/2-o(1)}$. 

$4$-Clique is known to be solvable in $O(n^{3.252})$ time, and if the matrix multiplication exponent $\omega$ is $>2$, the best running time for $4$-clique would still be truly supercubic.  Thus, the update time in our conditional lower bound, $m^{(c-2)/2-o(1)}$ is polynomially better than $m^{0.5-o(1)}$, as long as $\omega>2$. Recent results \cite{AlmanVW18,Alman19} show that the known techniques for matrix multiplication cannot show that $\omega$ is less than $2.16$.
%, and so $\omega>2$ is now plausible.

While the connection between clique detection and $s$-$t$ SP is interesting in its own right, it does not resolve the limitation on the preprocessing time of our previous lower bound.
To fix this, we introduce an online version of $4$-Clique, generalizing the OMv (actually the related OuMv \cite{henzinger2015unifying}) problem: 

\begin{definition}[OMv3 problem] \label{OMv3}
In the OMv3 problem, we are given an $n\times n$ Boolean matrix $A$ that can be preprocessed and then $n$ queries consisting of three length $n$ Boolean vectors $u,v,w$ have to be answered online by outputting the Boolean value 
$$\bigvee_{i,j,k} (u_i \wedge v_j \wedge w_k\wedge A[i,j] \wedge A[j,k]\wedge A[k,i]).$$
\end{definition}

One can think of $u,v,w$ as giving the neighbors of an incoming vertex $q$ in the three partitions of a tripartite graph, and then the Boolean value just answers whether $q$ would be part of a $4$-Clique if it were added to the graph. 
This is the natural extension of Henzinger et al.'s OuMv problem. OMv3 is easy to solve in $O(n^\omega)$ time per query by computing whether the neighborhood defined by $u,v,w$ contains a triangle. We hypothesize that there is no better algorithm, even if one is to preprocess the matrix in arbitrary polynomial time:

\begin{hypothesis}[OMv3 Hypothesis]\label{hyp:OMv3}
Any algorithm solving OMv3 with polynomial preprocessing time needs $n^{\omega+1-o(1)}$ total time to solve OMv3 in the word-RAM model with $O(\log n)$ bit words.
\end{hypothesis}

Using this Hypothesis, using essentially the same reduction as from $4$-Clique to $s$-$t$ SP, we obtain plausible conditional lower bounds for arbitrary polynomial preprocessing time and polynomially higher than $m^{0.5-o(1)}$ update/query time lower bound, improving the prior known results.

\begin{theorem} \label{thm:hardnessOMv}
In the word-RAM model with $O(\log m)$ bit words, under Hypothesis \ref{hyp:OMv3}, 
any incremental/decremental $s$-$t$ Shortest Paths algorithm with polynomial preprocessing time needs $m^{(\omega-1)/2-o(1)}$ amortized update or query time. For the current value of $\omega$, the update lower bound is $\Omega(m^{0.626})$.
\end{theorem}

In terms of both $m$ and $n$, Theorem \ref{thm:hardnessOMv} implies that when $m=O(n)$, partially dynamic $s$-$t$ Shortest Paths with arbitrary polynomial preprocessing needs total time $mn^{(\omega-1)/2-o(1)}$. This is the best limitation to date that both allows for arbitrary polynomial preprocessing and also holds for sparse graphs. If one considers ``combinatorial'' algorithms (i.e. where $\omega=3$), one gets that ES trees are essentially optimal again.

% We point out that our lower bound construction also captures 4-Clique and provides an interesting connection between these two problems. It is straight-forward to change the hypothesis used in Theorem \ref{thm:hardnessOMv} to the $4$-Clique Hypothesis stated below, resulting in a lower bound of $m^{0.626 - o(1)}$. For both lower bounds, we use the current values for $\omega$ and $\delta$, thus if one of our lower bound is refuted either the fastest algorithm for Matrix Multiplication or $4$-Clique is improved. 

We refer the reader to section \ref{sec:fineGrained} for a more detailed discussion of our fine-grained results, including further discussion of the plausibility of our conjectures.

% \begin{hypothesis}[$4$-Clique Hypothesis]
% There is a $\delta > 0$, such that no algorithm can detect a $4$-Clique in time $n^{3+\delta}$.
% \end{hypothesis}

%todo move related work here for arxiv submission

%% file: lowerbounds.tex
\section{Fine-grained lower bounds for partially dynamic s-t Shortest Paths}
\label{sec:fineGrained}

In this section we present several conditional lower bounds for the $s$-$t$ Shortest Paths problem in the partially dynamic, i.e. incremental or decremental, setting.
In the incremental setting, the assumption is that one starts with an empty graph and $m$ edges are inserted one by one. In the decremental setting, one is given an initial $m$-edge graph, and then its edges are deleted one by one in some order until the empty graph is reached. We will assume that no preprocessing is done, and that all the work of the algorithm is done in the updates and queries, however, in some cases we will be able to allow arbitrary polynomial preprocessing time. Our lower bounds are based on several popular hypotheses. All hypotheses are for the Word-RAM model of computation with $O(\log n)$ bit words.

The first, the {\em BMM Hypothesis} is a hypothesis about ``combinatorial'' algorithms, simple algorithms that do not use the heavy machinery of fast matrix multiplication (as in \cite{Sch81,Co97,cw90,vstoc12,dstothers,legallmmult}). The hypothesis (see e.g. \cite{popularconj,vsurvey}) states that in the Word-RAM model with $O(\log n)$ bit words, any combinatorial algorithm for computing the product of two $n\times n$ Boolean matrices requires $n^{3-o(1)}$ time. Due to the subcubic fine-grained equivalence of Boolean Matrix Multiplication (BMM) and Triangle detection \cite{williams2018subcubic}, the hypothesis is equivalent to: any combinatorial algorithm for Triangle detection in $n$-vertex graphs requires $n^{3-o(1)}$ time in the Word-RAM model of computation with $O(\log n)$ bit words.

A generalization of the BMM hypothesis is the Combinatorial $k$-Clique Hypothesis for constant $k\geq 3$ that asserts that any combinatorial algorithm for $k$-Clique detection in $n$-vertex graphs requires $n^{k-o(1)}$ time in the Word-RAM model of computation with $O(\log n)$ bit words. When one removes the restriction to combinatorial algorithms, the $k$-Clique Hypothesis becomes that the current fastest $k$-Clique algorithms are essentially optimal. For $k$ divisible by $3$, the assertion is that $n^{\omega k/3-o(1)}$ time is necessary (see \cite{AbboudBW15a,AbboudBW18,BringmannGMW18} for examples where this hypothesis is used). 

For $k$ not divisible by $3$, the best known running times for $k$-Clique are not as clean. For instance, for $4$-Clique the fastest known running time is $O(n^{3.252})$ using the fastest known rectangular matrix multiplication algorithm by Le Gall and Urrutia~\cite{legallurrutia}. As long as $\omega>2$, this algorithm would run in $O(n^{3+\delta})$ time for some $\delta>0$, i.e. in truly supercubic time.
Thus, the following quite weak $4$-Clique Hypothesis would be quite plausible: There is a $\delta>0$ so that $n^{3+\delta-o(1)}$ is needed to detect a $4$-Clique in an $n$-node graph. Looking at the current best $4$-Clique algorithms, of course, the $4$-Clique Hypothesis is plausible even for $\delta=0.252$. 

Another way to circumvent the ``combinatorial'' nature of the BMM Hypothesis when using it for lower bounds on dynamic algorithms, is to instead use the Online Matrix Vector Multiplication (OMv) Hypothesis of Henzinger et al. \cite{henzinger2015unifying}. The OMv Hypothesis is: Given an $n\times n$ Boolean matrix $A$, any algorithm that preprocesses $A$ in poly$(n)$ time needs total $n^{3-o(1)}$ time to answer $n$ online queries that give a length $n$ Boolean vector $v$ and ask for the Boolean product $Av$. The OMv Hypothesis is known to imply the related so called OuMv Hypothesis: Given an $n\times n$ Boolean matrix $A$, any algorithm that preprocesses $A$ in poly$(n)$ time needs total $n^{3-o(1)}$ time to answer $n$ online queries
$(u,v)$ where $u$ and $v$ are length $n$ Boolean vectors by returning the Boolean product $u^TA v$ right after $(u,v)$ is given.

We can generalize OuMv to define an analogous problem capturing $4$-Clique. Define OMv3 to be the following problem: Given an $n\times n$ Boolean matrix $A$, preprocess it so that $n$ queries of the following form can be answered online: the queries consist of three $n$ length Boolean vectors $u,v,w$, and the answer of the query should be the Boolean value 
$$\bigvee_{i,j,k} (u_i \wedge v_j \wedge w_k\wedge A[i,j] \wedge A[j,k]\wedge A[k,i]).$$

It is not hard to reduce $4$-Clique to OMv3, even when the queries are given offline: we can assume that $4$-Clique is given on a $4$-partite graph with partitions $V_1,V_2,V_3,V_4$. Let $A$ be the adjacency matrix of the subgraph induced by $V_1,V_2,V_3$, and for each $x\in V_4$, we can define the three Boolean vectors $u^x,v^x,w^x$, where $u^x[j]=1$ only if $x\in V_1$ and $(x,j)$ is an edge,  $v^x[j]=1$ only if $x\in V_2$ and $(x,j)$ is an edge, and  $w^x[j]=1$ only if $x\in V_3$ and $(x,j)$ is an edge. Then $(u^x_i \wedge v^x_j \wedge w^x_k\wedge A[i,j] \wedge A[j,k]\wedge A[k,i])= 1$ only when $(i,j,k,x)\in V_1\times V_2\times V_3\times V_4$ and $(i,x),(j,x),(k,x),(i,j),(j,k),(i,k)$ are all edges, i.e. whenever $(i,j,k,x)$ is a $4$-Clique.

Now, similarly to OuMv, since the queries to OMv3 are given in an online fashion, the problem seems harder than $4$-Clique.  The simple way to solve the problem, when given $u,v,w$ seems to be to take the submatrices $A^1,A^2,A^3$ where $A^1$ restricts to the rows that $u$ is one and columns that $v$ is one,  $A^2$ restricts to the rows that $v$ is one and columns that $w$ is one and $A^3$ restricts to the rows that $w$ is one and columns that $u$ is one, and then to compute the trace of $A^1\cdot A^2\cdot A^3$ in $O(n^\omega)$ time. In particular, there seems to be no way to use the fact that rectangular matrix multiplication can be done faster than by splitting into square blocks and using the fast square matrix multiplication algorithms.

We can thus make the following very plausible OMv3 Hypothesis, similar to the OuMv one, that any algorithm with polynomial preprocessing time needs $n^{\omega+1-o(1)}$ total time to solve OMv3.

The last hypothesis we will use concerns the $k$-Cycle problem (for constant $k$): given an $m$-edge graph, determine whether it contains a cycle on $k$ vertices.
All known algorithms for detecting $k$-cycles in directed graphs with $m$ edges run at best in time $m^{2-c/k}$ for various small constants $c$ \cite{YuZw04,AlYuZw97,lincolnsoda18,patternscycles19}, even using powerful tools such as fast matrix multiplication. Ancona et al.~\cite{ancona2019} formulated
a natural hypothesis completely consistent with the state of the art of cycle detection. This {\em $k$-Cycle Hypothesis} states that (in the Word-RAM model), for every constant $\eps>0$, there exists a constant $k$, so that there is no $O(m^{2-\eps})$ time algorithm that can find a $k$-cycle in an $m$-edge graph.

\subsection{Hardness from $k$-Cycle}
The {\em $k$-Cycle Hypothesis} states that (in the Word-RAM model), for every constant $\eps>0$, there exists a constant $k$, so that there is no $O(m^{2-\eps})$ time algorithm that can find a $k$-cycle in an $m$-edge graph.

% one needs $m^{2-f(k)-o(1)}$ time to find a $k$-cycle, for some continuous (over the reals) $f(k)$ that goes to $0$ as $k$ goes to infinity. In other words, for every constant $\eps>0$, there exists an integer $k$, such that detecting a $k$-cycle in an $m$-edge graph requires $m^{1-\eps-o(1)}$ time.

We will reduce $k$-Cycle in $m$-edge, $n$-node graphs to incremental $s$-$t$ SP in undirected or directed graphs, where one starts with an empty graph and inserts $O(m)$ edges, performing $O(n)$ queries. As $n=O(m)$ in connected graphs, the $k$-Cycle Hypothesis implies that as $k$ grows,
the amortized update/query time must be at least $m^{1-o(1)}$.

\begin{theorem}\label{thm:kcycred}
In the word-RAM model with $O(\log m)$ bit words, under the $k$-Cycle Hypothesis,  there can be no constant $\eps>0$ such that incremental $s$-$t$ SP in directed or undirected $m$-edge graphs can be solved with $O(m^{2-\eps})$ preprocessing time and $O(m^{1-\eps})$ amortized update and query time.
\end{theorem}

An analogous theorem holds in the decremental setting. We omit the details, but essentially one runs the reduction in reverse.

We note that with very minor modification, our reduction can be made to go from {\em minimum weight} $k$-cycle to incremental or decremental shortest $s$-$t$ path in weighted graphs. Lincoln et al.~\cite{lincolnsoda18} showed that under very believable assumptions (that min weight $k$-clique and also clique in hypergraphs require $n^{k-o(1)}$ time), min weight $k$-cycle requires $m^{2-1/k-o(1)}$ time, and hence Theorem~\ref{thm:kcycred} holds under even more standard assumptions for weighted graphs. For unweighted graphs, we do need the unweighted $k$-Cycle assumption. Even though this assumption has so far not been related to other standard hardness hypotheses, it is believable and completely consistent with the current state of algorithms.

We will now prove Theorem~\ref{thm:kcycred}. The reduction is the natural extension of the reduction from Triangle detection to $s$-$t$ SP in \cite{popularconj}. See Figure \ref{fig:kcyc}.

\begin{figure}[h]
  \centering
  \includegraphics[width=.7\linewidth]{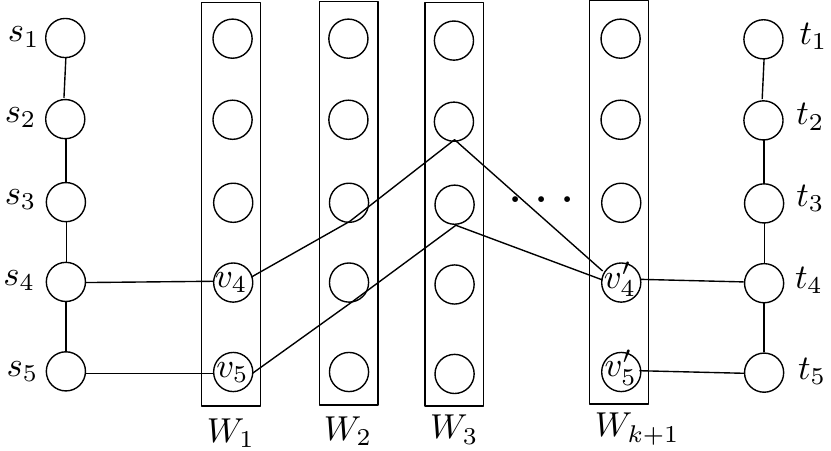}
  \caption{The edges between $W_i$ and $W_{i+1}$ are the edges of the original graph between $V_i$ and $V_{i+1}$. The figure shows the state of the dynamic graph at stage $4$ when one is searching for a $k$-cycle including $v_4$. A path from $s_1$ to $t_1$ that goes through edge $(s_5,v_5)$ instead of $(s_4,v_4)$ will have length at least $9+k$, whereas if there is a $k$-cycle through $v_4$, the shortest path will use $(s_4,v_4)$ and $(v'_4,t_4)$ and will have length $8+k$.}
  \label{fig:kcyc}
\end{figure}

First, suppose that incremental $s$-$t$ SP can be solved with $O(m^{2-\eps})$ preprocessing time and $O(m^{1-\eps})$ update and query time for some constant $\eps>0$. For that $\eps$, let $k$ be such that the $k$-Cycle Hypothesis asserts that there is no $O(m^{2-\eps})$ time algorithm for $k$-Cycle in $m$-edge graphs. We will obtain a contradiction via our reduction.

Let $G$ be an $m$-edge, $n$-vertex graph in which we want to find a $k$-cycle. First we use color-coding \cite{colorcoding,colorcodingj,colorcodingenc} so that with polylogarithmic time overhead, we can assume that the vertices of $G$ are partitioned into $V_1,V_2,\ldots,V_k$, so that if $G$ contains a $k$-Cycle, one such cycle has its $i$th vertex in $V_i$, for each $i\in \{1,\ldots,k\}$.

Now, the vertices of our incremental graph will be as follows:
\begin{itemize}
\item For every $i\in \{1,\ldots,k\}$, there is a set of vertices $W_i$ that contains for every $v\in V_i$  a vertex $v\in W_i$ representing it (slight abuse of notation here).
\item Another copy of the vertices of $V_1$ in a set $W_{k+1}$. Call the copy of $v\in V_1$ in $W_{k+1}$, $v'$.

\item A source vertex $s_1$, followed by vertices $s_1,\ldots, s_n$, all connected in a path.
\item A sink vertex $t_1$, preceded by vertices $t_n,\ldots,t_{2}$, all connected in a path $t_n\rightarrow t_{n-1}\rightarrow \ldots \rightarrow t_1$.
\end{itemize}

Besides the path edges above, the remaining edges to be inserted before any queries are as follows:
 For every $i\in \{1,\ldots,k\}$, for every $u\in W_i,v\in W_{i+1}$, insert $(u,v)$ as an edge if $(u,v)$ was an edge of $G$.

Notice now that due to the color-coding, we can assume that to detect a $k$-Cycle in $G$ we only need to check whether for some $v\in W_1$ and its copy $v'\in W_{k+1}$ there is a path of length $k$. Because of the layering, the distance between $v$ and $v'$ is $k$ if there is a $k$-cycle in $G$ going through $v$ and it is $> k$ otherwise.

Now, the rest of the dynamic stages proceed as follows. Let the vertices of $W_1$ be $v_1,\ldots,v_n$, and their corresponding copies in $W_{k+1}$ be $v'_1,\ldots,v'_n$. The stages go from $1$ to $n$. In stage $i$, we insert an edge between $s_{n+1-i}$ and $v_{n+1-i}\in W_1$ and an edge between $v'_{n+1-i}$ and $t_{n+1-i}$. Then we query the distance between $s_1$ and $t_1$.

Now, at stage $i$, we have edges between $s_{n+1-j}$ and $v_{n+1-j}$ and between $t_{n+1-j}$ and $v'_{n+1-j}$ for all $j\in \{1,\ldots,i\}$.

The shortest path from $s_1$ to $t_1$ looks like this: go from $s_1$ to $s_{n+1-j}$ (for some $j\in \{1,\ldots,i\}$) using the path of $s$-nodes, then take an edge to $W_1$, go through the layers $W_1-W_{k+1}$ to a node of $W_{k+1}$ and then to $t_{n+1-r}$ (for some $r\in \{1,\ldots,i\}$) and then to $t_1$. Since going from a vertex in $W_{1}$ to a vertex in $W_{k+1}$ gives distance at least $k$, the length of this path is at least $(n-j)+2+k+(n-r)$. 
If one of $j$ or $r$ is not equal to $i$ (i.e. it is $<i$), the length of the path is $>2(n-i)+k+2$. If $G$ contains a $k$-Cycle through $v_{n+1-i}$, however, there is a path from $s_1$ to $t_1$ going through $v_{n+1-i}, v'_{n+1-i}$ and the edges of the $k$-Cycle, having total length exactly $2(n-i)+k+2$. 

Thus, in stage $i$, the distance between $s_1$ and $t_1$ is $2(n-i)+k+2$ if there is a $k$-Cycle through $v_{n+1-i}$, and otherwise the distance is $>2(n-i)+k+2$.
The total number of edge insertions is $O(m+n)=O(m)$ and the number of queries is $O(n)=O(m)$. Thus our supposedly efficient incremental algorithm would solve the $k$-Cycle problem in time $\tilde{O}(m\cdot m^{1-\eps})=\tilde{O}(m^{2-\eps})$ time, a contradiction.

\subsection{Hardness from OMv3 and 4-Clique}
A weakness of the reduction from $k$-Cycle detection is that the result is only meaningful when the preprocessing time used by the algorithms is $O(m^{2-\eps})$ for some $\eps>0$. For incremental algorithms, one could argue that since one starts with an empty graph, it is unclear how preprocessing could help at all. For decremental graphs however, one knows all the edges so preprocessing could help.

We present a reduction from OMv3 that (1) allows for arbitrary polynomial preprocessing, and (2) gives a higher conditional lower bound than the $m^{0.5-o(1)}$ amortized update/query lower bound that follows from OMv \cite{henzinger2015unifying}, as long as $\omega>2$.

Moreover, even if instead of the OMv3 Hypothesis we only use the $4$-Clique Hypothesis (still via the same reduction below as OMv3), we still obtain a higher than $m^{0.5-o(1)}$ update lower bound.

We will prove Theorem~\ref{thm:OMv3tost} below.

\begin{theorem}\label{thm:OMv3tost}
Suppose that incremental or decremental $s$-$t$ Shortest Paths can be maintained with $P(m)$ preprocessing time and $u(m)$ amortized update and query time, then OMv3 can be solved with $P(O(n^2))$ preprocessing time and $n^2\cdot u(O(n^2))$ total query time.
\end{theorem}

If we assume the OMv3 Hypothesis, then we obtain that any incremental/decremental $s$-$t$ Shortest paths algorithm with polynomial preprocessing time needs $m^{(\omega-1)/2-o(1)}$ amortized update or query time. For the current value of $\omega$, the update lower bound is $\Omega(m^{0.686})$.

If we assume the $4$-Clique hypothesis we obtain that there exists a $\delta>0$ such that any incremental/decremental $s$-$t$ Shortest paths algorithm needs either $m^{(3+\delta)/2-o(1)}$ preprocessing time, or $m^{(1+\delta)/2-o(1)}$ amortized update or query time. For the current value of $\delta$, this update lower bound is $\Omega(m^{0.626})$.
%Both update time lower bounds are already higher than the $m^{1/2-o(1)}$ lower bounds from prior work \cite{omv,abboudv}.

We now begin the proof of Theorem~\ref{thm:OMv3tost}.
We begin with a gadget that encodes the row/column indices $i\in [n]$ of $A$ and dynamically encodes the $n$ queries $(u^1,v^1,w^1),\ldots,(u^n,v^n,w^n)$ when they come.

We will describe the gadget $G(u)$ which will encode the $u^i$s.  The gadgets $G(v)$ and $G(w)$ are analogous. See Figure~\ref{fig:inner}.
$G(u)$ consists of $n(n+1)$ vertices: every $i\in [n]$ gets $n+1$ copies $(i,0),(i,1),\ldots, (i,n)$. The vertices $(i,1),\ldots,(i,n)$ for each particular $i$ are chained together in a path, so that for all $j\in \{1,\ldots,n-1\}$ there is an edge between $(i,j)$ and $(i,j+1)$. This describes $G(u)$ (and also $G(v),G(w)$) before any queries come. On query $u^\ell$, one inserts an edge from $(i,0)$ to $(i,\ell)$ for each $i$ for which $u^\ell[i]=1$. (The insertions for $G(v)$ and $G(w)$ are analogous but with $v^\ell$ and $w^\ell$ instead of $u^\ell$, respectively.)

 \begin{figure}[h]
  \centering
  \includegraphics[width=.7\linewidth]{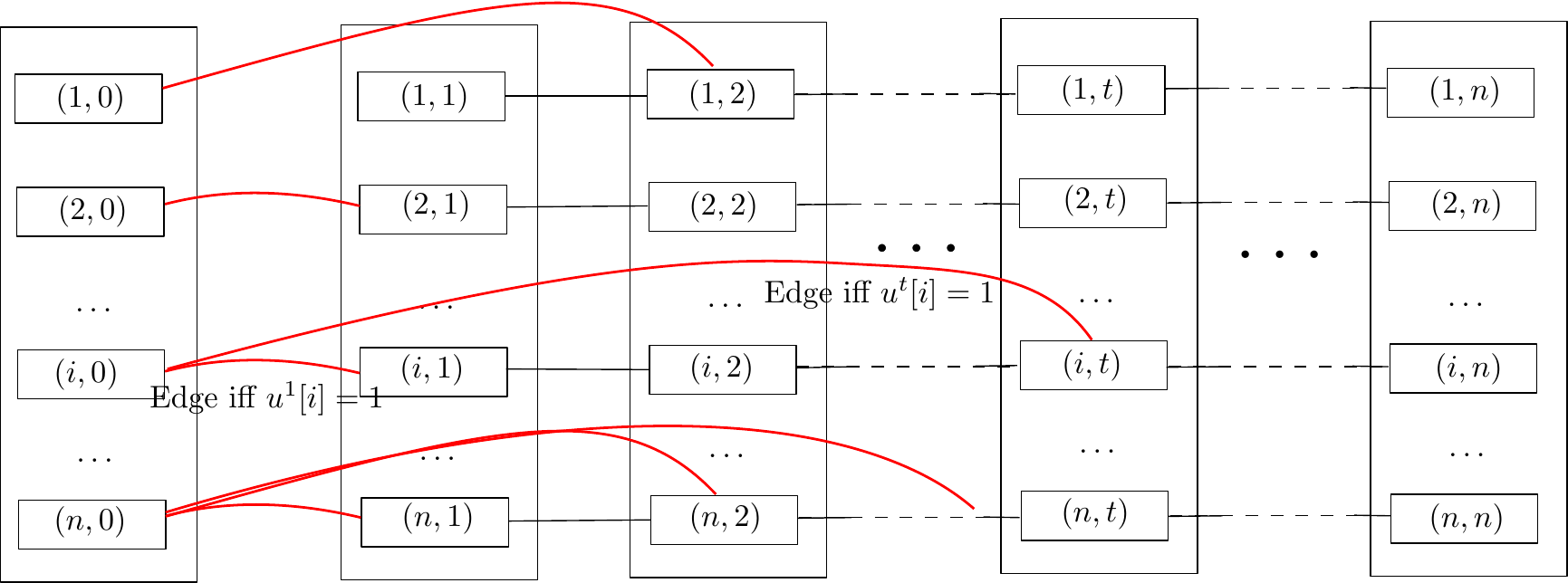}
  \caption{The gadget $G(u)$ encoding the positions in which the queries $u^j$ are $1$; in particular, if the current query is $u^t$, for each $j\leq t$ and each $i\in[n]$ there is a red edge from $(i,0)$ to $(i,j)$ whenever $u^j[i]=1$, and the edges for $u^t$ are inserted right after $u^t$ is queried.}
  \label{fig:inner}
\end{figure}

There are two copies of $G(u)$, $G(u)$ and $G'(u)$. We chain $G(u),G(v),G(w),G'(u)$ together as follows. For every $i,j$ such that $A[i,j]=1$ we add edges from $(i,n)$ of $G(u)$ to $(j,0)$ of $G(v)$, from $(i,n)$ of $G(v)$ to $(j,0)$ of $G(w)$, and from $(i,n)$ of $G(w)$ to $(j,0)$ of $G'(u)$. See Figure~\ref{fig:middle}.

 \begin{figure}[h]
  \centering
  \includegraphics[width=.9\linewidth]{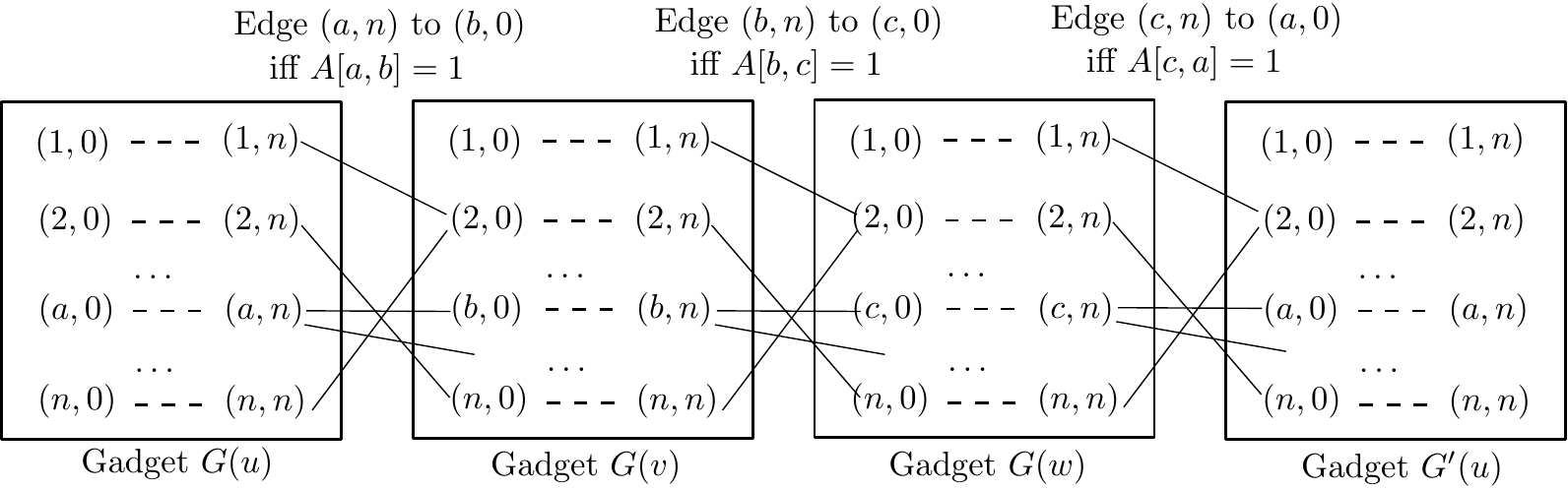}
  \caption{The ``middle'' gadget connecting the gadgets $G(u),G(v),G(w),G'(u)$.}
  \label{fig:middle}
\end{figure}

Notice that so far we have $O(n^2)$ vertices and edges. 
\begin{claim}\label{claim:1}
Right after inserting the edges for $u^\ell,v^\ell,w^\ell$ into $G(u),G(v),G(w),G'(u)$, the distance between $(i,0)$ in $G(u)$ and $(i,n)$ in $G'(u)$ is $3+4(n+1-\ell)$ if there are some $j,k$ so that $$u^\ell[i] \wedge v^\ell[j]\wedge w^\ell[k] \wedge A[i,j]\wedge A[j,k]\wedge [k,i]=1,$$ and the distance is $>3+4(n+1-\ell)$ otherwise.
\end{claim}
\begin{proof}
To get from $G(u)$ to $G'(u)$ one needs to use at least $3$ edges (between $G(u)$ and $G(v)$, between $G(v)$ and $G(w)$ and between $G(w)$ and between $G'(u)$) and then also one needs to go from layer $0$ ($(*,0)$) to layer $n$ ($(*,n)$) in each of the $4$ gadgets. The shortest possible way to do this is to go through an edge from $(j,0)$ to $(j,\ell)$ and then along the path from $(j,\ell)$ to $(j,n)$, altogether having length $n+1-\ell$. Thus the shortest a path from $(i,0)$ in $G(u)$ to $(i,n)$ in $G'(u)$ is $3+4(n+1-\ell)$. 

This minimal length is achievable if and only if (1) there are some $j$ and $k$ so that the edges $(i,0)$ to $(i,\ell)$, $(j,0)$ to $(j,\ell)$ and $(k,0)$ to $(k,\ell)$ exist in $G(u),G'(u)$ and $G(v)$ and $G(w)$, respectively, and (2) also the edges $(i,n)$ to $(j,0)$ from $G(u)$ to $G(v)$, $(j,n)$ to $(k,0)$ from $G(v)$ to $G(w)$ and $(k,n)$ to $(i,0)$ from $G(w)$ to $G'(u)$ also exist. That is, if and only if $u^\ell[i] \wedge v^\ell[j]\wedge w^\ell[k] \wedge A[i,j]\wedge A[j,k]\wedge [k,i]=1.$
\end{proof}

We will now complete the construction.
Beyond the gadgets $G(u),G'(u),G(v),G(w)$ and the connections between them, we add two paths:
\begin{itemize}
\item The first consists of vertices $s_{\ell,i}$ for $\ell,i\in [n]$ and edges $(s_{\ell,i},s_{\ell,i+1})$ when $i<n$ and $(s_{\ell,n},s_{\ell+1,1})$ for $\ell<n$. 

\item The second similarly consists of vertices $t_{\ell,i}$ for $\ell,i\in [n]$ and edges $(t_{\ell,i},t_{\ell,i+1})$ when $i<n$ and $(t_{\ell,n},t_{\ell+1,1})$ for $\ell<n$.
\end{itemize}

The source and sink for the s-t shortest paths instance are $s_{n,n}$ and $t_{n,n}$.

On query $(u^j,v^j,w^j)$ to OMv3, we insert the already described edges into $G(u),G'(u),G(v),G(w)$ and then perform the following $n$ insertions and queries: For each $a$ from $1$ to $n$, insert the edges $s_{j,a}$ to $(a,0)$ in $G(u)$ and $t_{j,a}$ to $(a,n)$ in $G'(u)$; then query the distance between $s_{n,n}$ and $t_{n,n}$. See Figure~\ref{fig:all}.

\begin{figure}[h]
  \centering
  \includegraphics[width=.8\linewidth]{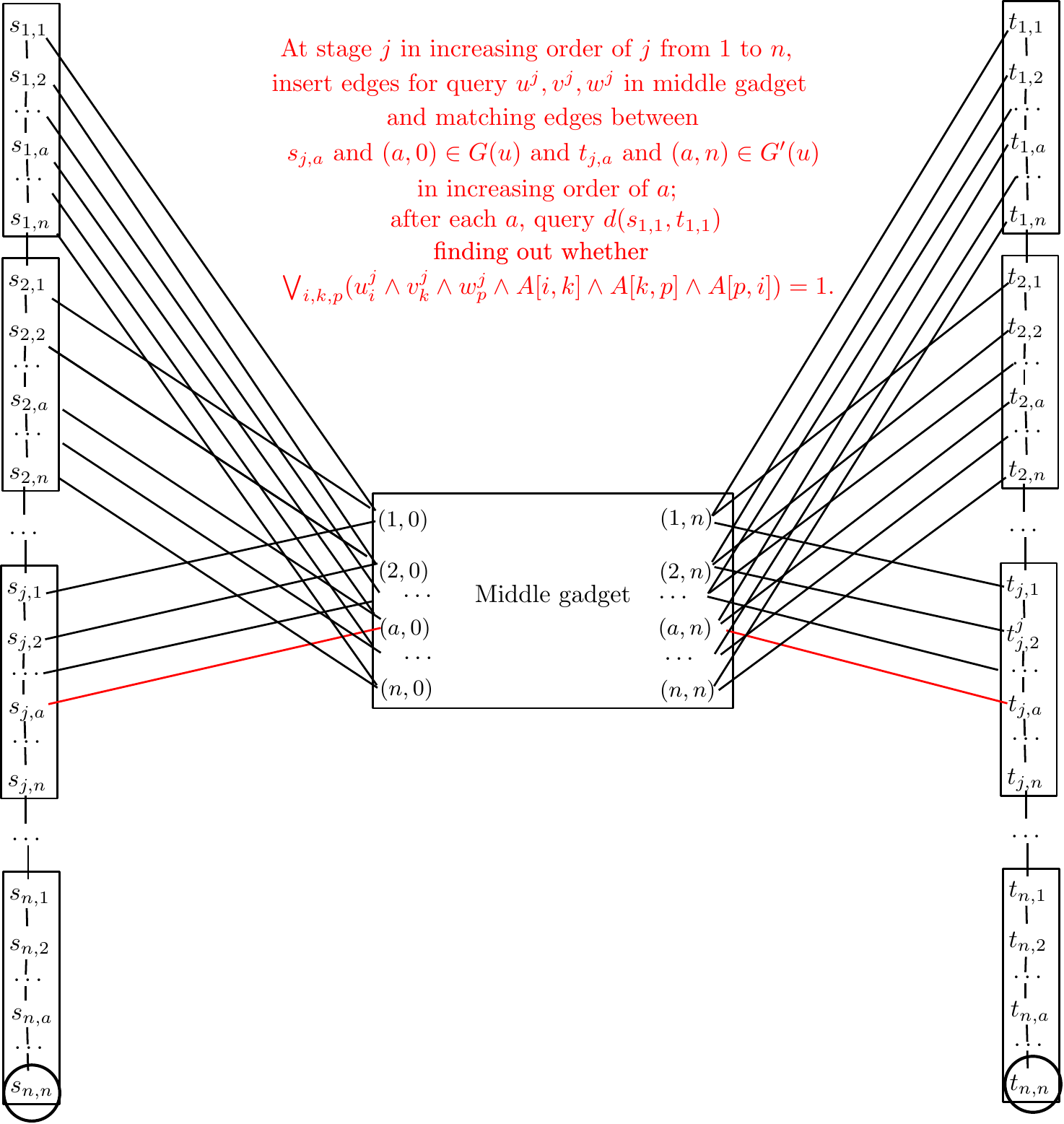}
  \caption{The full reduction connecting the middle gadget with the source and sink paths.}
  \label{fig:all}
\end{figure}

\begin{claim}\label{claim:2}
After inserting the edges $s_{j,a}$ to $(a,0)$ in $G(u)$ and $t_{j,a}$ to $(a,n)$ in $G'(u)$, the distance from $s_{n,n}$ to $t_{n,n}$ is $2(n-j)n+2(n-a+1)+3+4(n+1-j)$ if there are some $b,c$ so that $$u^j[a] \wedge v^j[b]\wedge w^j[b] \wedge A[a,b]\wedge A[b,c]\wedge A[c,a]=1,$$ and the distance is $>2(n-j)n+2(n-a+1)+3+4(n+1-j)$ otherwise.
\end{claim}

\begin{proof}
The shortest path from $s_{n,n}$ to $t_{n,n}$ goes from $s_{n,n}$ up the $s$-path (on the left in Figure~\ref{fig:all}) to some node $s^p_b$, then along the edge $(s^p_b,(b,0))$ to the first layer in the Middle gadget, then through the middle gadget, exiting it at some node $(c,n)$ in the last layer, going to a node $(r,c)$ on the $t$-path (on the right in Figure~\ref{fig:all}) down to $t_{n,n}$.

%Since the edges on the s- and t- paths are subdivided to have length $2$, t
The length of this path is the length of the subpath from $(b,0)$ to $(c,b)$ inside the middle gadget $+$ \[n(n-p)+(n-b)+n(n-r)+(n-c)+2.\]

By Claim~\ref{claim:1}, the shortest possible distance between the first and last layers of the Middle gadget, after inserting the edges inside it for $u^j,v^j,w^j$ is $3+4(n+1-j)$.
Thus, for a particular choice of $p,r\leq j$ and $b,c\in[n]$, the length of the above path is at least 
\[3+4(n+1-j)+n(n-p)+(n-b)+n(n-r)+(n-c)+2.\]

If $p$ and $r$ are both $\leq j-1$, the length of the path is at least:
\[3+4(n+1-j)+2n(n-j)+2n+(n-b)+(n-c)+2> 3+4(n+1-j)+2n(n-j)+2(n-a)+2,\]
since $b,c\leq n$ and $a\geq 1$.

If $p\leq j-1$ and $r=j$ (the case $p=j$ and $r\leq j-1$ is similar), then since the only added edges from the $j$th part of the $t$ path are between $(c,n)$ and $t_{j,c}$ for $c\leq a$ at this point, we also get that the length of the path is at least \[3+4(n+1-j)+2n(n-j)+n+(n-a)+2> 3+4(n+1-j)+2n(n-j)+2(n-a)+2.\]

Similarly, if $p,r=j$, and $b$ or $c$ is $<a$, the length of the path is at least 
\[3+4(n+1-j)+2n(n-j)+2(n-a)+1+2> 3+4(n+1-j)+2n(n-j)+2(n-a)+2.\]
Finally, by Claim~\ref{claim:1}, if $p=r=j,b=c=a$, then the length of the path is $3+4(n+1-j)+2n(n-j)+2(n-a)+2$ if there are some $b,c$ such that $u^j[a]\wedge v^j[b]\wedge w^j[c]\wedge A[a,b]\wedge A[b,c]\wedge A[c,a]=1$, and the length is larger otherwise.
\end{proof}

%% file: app.tex
\section{Related Work}\label{app}

\paragraph{Reachability, Strongly-Connected Components and Topological Order.} The problems most related to SSSP in directed graphs are the easier problems of maintaining single-source reachability, strongly-connected components and the topological order of the graph. In decremental graphs, these problems have recently be solved to near-optimality \cite{bernstein2019decremental} following a long line of research \cite{shiloach1981line, italiano1988finding, roditty2016fully, lkacki2011improved, chechik2016decremental, italiano2017decremental} whilst in incremental graphs even the complexity of cycle detection is still open with the currently best bounds implying total update time $\tilde{O}(\min\{m^{4/3}, m\sqrt{n}, n^2\})$ \cite{bhattacharya2018improved, bernstein2018incremental, bender2016new}. Thus, it is conceivable that the incremental SSSP problem might be no easier than its decremental counterpart. Finally, we point out that the problem of all-pairs reachability, often referred to as transitive closure, has also been considered and solved to near-optimality in decremental and fully dynamic graphs \cite{lkacki2011improved, roditty2016fully}.

\paragraph{Single Source Shortest Paths in Undirected Graphs.} In undirected graphs, Bernstein and Roditty \cite{bernstein2011improved} gave the first improvement over the classic ES-tree data structure \cite{shiloach1981line} by presenting an algorithm for decremental unweighted $(1+\eps)$-approximate SSSP with total time $n^2 2^{O(\sqrt{\log (n)})}$. It was subsequently shown by Henzinger et al.~\cite{henzinger2017sublinear, HenzingerKN14} that subquadratic update time was possible and they then gave an approach \cite{henzinger2014decremental} with total update time $m^{1+O(\log^{5/4} ((\log n)/\eps))/\log^{1/4} n}\log W=m^{1+o(1)}\log W$ for weighted graphs which is also believed to work in the incremental setting (though not explicitly stated). These data structures, however, are all randomized and assume an oblivious adversary. Consequently, they can not be used as a black-box in many applications. 

To address this issue, Bernstein and Chechik gave the first deterministic partially dynamic $(1+\eps)$-approximate algorithms that improve upon the ES-tree data structure. They first presented a data structure with total update time $\tilde{O}(n^2)$ \cite{bernstein2016deterministic} which was extended to handle weights in total time $\tilde{O}(n^2 \log W)$ \cite{bernstein2017deterministicweighted}. Their data structures, however, work by contracting dense parts of graphs and they are therefore not able to output corresponding paths but only distance estimates\footnote{In the conference version of \cite{bernstein2016deterministic}, the authors claim that there data structure can be extended to return shortest paths but that they defer the proof to the full version, however, in \cite{bernstein2017deterministicweighted} one of the authors points out that this issue could not be resolved. }. Very recently, this issue was addressed by Chuzhoy and Khanna \cite{Chuzhoy:2019:NAD:3313276.3316320} who gave an algorithm with $n^{2+o(1)} \log W$ update time under vertex deletions that works against an adaptive adversary and that can return approximate shortest paths. Using their new data structure, they then showed that various flow and cut problems can be improved using this data structure in a black box fashion. However, their algorithm only works assuming vertex deletions and requires $n^{1+o(1)}$ query time for a path.  Further, Bernstein and Chechik recently gave an algorithm with total update time $\tilde{O}(mn^{3/4})$ \cite{bernstein2017deterministic} that also improves the running time in unweighted sparse graphs which in turn was improved to $O(mn^{0.5+o(1)})$ by Probst Gutenberg and Wulff-Nilsen \cite{DetDecrementalSSSP}.

Finally, we point out that for the fully dynamic SSSP problem a trivial lower bound is implied by the APSP conjecture, since with $O(n)$ updates and $O(n^2)$ queries, the source can be added via a unit-weight edge to any vertex, and then the distances can be queried and the edge can be removed. This lower bound also extends to $(1+\eps)$-approximate SSSP since even static APSP with multiplicative stretch $(1+\eps)$ and additive stretch $<2$ is believed to have no truly subcubic algorithm under the BMM Hypothesis as well as other hypotheses \cite{cohen2001all, williams2010subcubic}.

\paragraph{All Pairs Shortest Path in Undirected Graphs}

For the undirected APSP problem in decremental graphs, Henzinger et al.~\cite{henzinger2014decremental} presented an algorithm with stretch $((2+\eps)^k - 1)$ and total update time $m^{1+1/k+o(1)} \log^2 W$ for any positive integer $k$. They also gave an algorithm with stretch $(2+\eps)$ or $(1+\eps, 2)$ and with total update time $\tilde{O}(n^{2.5})$ in \cite{henzinger2016dynamic} and a $(1+\eps)$-approximate deterministic algorithm with $\tilde{O}(mn/\eps)$ update time which derandomized the construction by Roditty and Zwick \cite{roditty2012dynamic} with matching running time. Recently, Chechik  \cite{chechik2018near} presented an algorithm with $(2+\eps)k-1$-approximate algorithm with update time $mn^{1/k + o(1)} \log W$ for any positive integer $k$ and constant $\eps$, whose total update time matches the preprocessing time of static distance oracles \cite{thorup2005approximate} with corresponding stretch.

For fully dynamic graphs, Bernstein \cite{bernstein2009fully} gave an algorithm with stretch $(2+\eps)$ with $m^{1+o(1)}$ update time which constitutes the first improvement over the data structure by Demetrescu and Italiano \cite{demetrescu2004new}. The first data structure for fully dynamic graphs and update time sublinear in $n$ and constant stretch was given by Abraham et al.~\cite{abraham2014fully}.

\paragraph{All Pairs Shortest Paths in Directed Graphs}

For the all-pairs shortest paths (APSP) problem, Baswana et al.~\cite{baswana2002improved, baswana2007improved} presented an algorithm for exact decremental unweighted graphs with total update time $\tilde{O}(n^3 n)$, and also gave a $(1+\eps)$ approximate version with total update time $\tilde{O}(n^2\sqrt{ (m/\eps)})$. The latter result was then improved to $\tilde{O}(mn \log W)$ total update time by Bernstein \cite{bernstein2016maintaining}. For fully dynamic graphs, a line of research \cite{king1999fully, demetrescu2006fully} culminated in a data structure by Demetrescu and Italiano \cite{demetrescu2004new} with $\tilde{O}(n^2)$ amortized update time that is exact, deterministic and works for weighted graphs with no dependence in $W$. Their data structure was slightly improved, generalized and simplified by Thorup \cite{thorup2004fully}. Further, Thorup \cite{thorup2005worst} gave a deterministic data structure for fully dynamic APSP with $\tilde{O}(n^{2+3/4})$ \emph{worst-case} update time which was recently improved using randomization to $\tilde{O}(n^{2+2/3})$ by Abraham et al.~\cite{abraham2017fully}, and to $O(n^{2.71})$ deterministically by Probst Gutenberg and Wulff-Nilsen \cite{sodaAPSP}.